\documentclass[letterpaper, 11pt]{article}

\usepackage[utf8]{inputenc}
\usepackage[margin=1in]{geometry}

\usepackage{amsfonts, amsmath, amssymb, amsthm, mathabx}
\usepackage{bm}
\usepackage{bbm}
\usepackage{soul}

\usepackage{multirow}

\usepackage[numbers, sort]{natbib}
\usepackage{algorithm2e}
\usepackage{algorithmic}
\usepackage{todonotes}
\usepackage{xcolor}

\usepackage[english]{babel}
\usepackage{todonotes}

\usepackage{comment}

\newtheorem{theorem}{Theorem}
\newtheorem{lemma}{Lemma}

\newtheorem{definition}{Definition}
\newtheorem*{conjecture*}{Conjecture}
\newtheorem*{informal}{Informal Theorem}

\newcommand{\zhiyi}[1]{\todo[inline, color=blue!20]{Zhiyi: #1}}

\newcommand{\opt}{\textsc{Opt}}

\renewcommand{\vec}[1]{\mathbf{#1}}
\renewcommand{\Pr}{\mathbf{Pr}}
\newcommand{\E}{\mathbb{E}}
\newcommand{\skl}{D_{\textsc{SKL}}}

\newcommand{\shading}{s_{N, n, \delta}}
\newcommand{\doubleshading}{d_{N, n, \delta}}

\newcommand{\Hellinger}{H}
\newcommand{\TV}{\delta}
\newcommand{\hel}{H^2}

\newcommand{\defeq}{\overset{\text{def}}{=}}

\newcommand{\hypo}{h}
\newcommand{\hypoclass}{\mathcal{H}}

\newcommand{\R}{\mathbb{R}}

\newcommand{\typespace}{T}
\newcommand{\type}{t}

\title{
    Generalizing Complex Hypotheses on Product Distributions:\\
    Auctions, Prophet Inequalities, and Pandora's Problem
}
\author{
    Chenghao Guo\thanks{IIIS, Tsinghua University. Email: \{guoch16, zhang-xz16\}@mails.tsinghua.edu.cn.}
    \and
    Zhiyi Huang\thanks{The Univeristy of Hong Kong. Email: zhiyi@cs.hku.hk.}
    \and
    Zhihao Gavin Tang\thanks{ITCS, Shanghai University of Finance and Economics. Email: tang.zhihao@mail.shufe.edu.cn.}
    \and
    Xinzhi Zhang\footnotemark[1]
}

\begin{document}

\begin{titlepage}
\thispagestyle{empty}
\maketitle
\begin{abstract}
    \thispagestyle{empty}
    This paper explores a theory of generalization for learning problems on product distributions, complementing the existing learning theories in the sense that it does not rely on any complexity measures of the hypothesis classes.
    The main contributions are two general sample complexity bounds:
    (1) $\tilde{O} \big( \frac{nk}{\epsilon^2} \big)$ samples are sufficient and necessary for learning an $\epsilon$-optimal hypothesis in \emph{any problem} on an $n$-dimensional product distribution, whose marginals have finite supports of sizes at most $k$;
    (2) $\tilde{O} \big( \frac{n}{\epsilon^2} \big)$ samples are sufficient and necessary for any problem on $n$-dimensional product distributions if it satisfies a notion of strong monotonicity from the algorithmic game theory literature.
    As applications of these theories, we match the optimal sample complexity for single-parameter revenue maximization (Guo et al., STOC 2019), improve the state-of-the-art for multi-parameter revenue maximization (Gonczarowski and Weinberg, FOCS 2018) and prophet inequality (Correa et al., EC 2019; Rubinstein et al., ITCS 2020), and provide the first and tight sample complexity bound for Pandora's problem.
\end{abstract}

\end{titlepage}

\section{Introduction}
\label{sec:introduction}

\begin{quote}
    \em
    The learning process is a process of choosing an appropriate function from a given set of functions.
    \hspace*{\fill} --- \citet{Vapnik/book/1998}
\end{quote}

\noindent
Generalization is widely recognized as one of the fundamental pillars of learning theory.
A general learning problem asks whether we can select a function, often referred to as a hypothesis, from a hypothesis class to maximize or minimize the expectation w.r.t.\ an underlying distribution over the data domain, based on samples from the distribution.
While it may be easy to select a hypothesis the maximizes or minimizes the average over the samples, how can we ensure that it \emph{generalizes} and gets a similar performance on the underlying distribution?
More quantitatively, we may ask about its sample complexity: 
\emph{how many samples are sufficient and necessary for choosing a hypothesis that is optimal on the true distribution up to an $\epsilon$ error?}

A widely studied example is the classification problem in supervised learning.
In this problem, each data point is a feature-label pair $(x, y) \in X \times Y$, where $X$ is the feature domain, e.g., $\R^n$, and $Y$ is the label domain, e.g., $\{0, 1\}$.
Each hypothesis corresponds to a classifier, i.e., a feature-to-label mapping $f : X \mapsto Y$;
its value on a data point $(x, y)$ is $L\big(f(x), y\big)$ for some loss function $L$, e.g., $\big| f(x) - y \big|$.
The goal is to learn a classifier from samples to minimize the expected loss on the underlying distribution.

Meanwhile, the general learning problem also captures a wide range of optimization problems in the Bayesian model.
The problem of learning revenue-maximizing auctions from data is a recent example, which has received a lot of attention in algorithmic game theory and more generally in theoretical computer science.
In this example, each data point comprises the valuations of the bidders;
each hypothesis corresponds to an auction and its value is defined to be the revenue of the auction on the given valuations.
We aim to learn an auction from sample valuations to maximize the expected revenue on the underlying value distributions.

\subsection{Generalization from Complexity Measures of the Hypothesis Class}
\label{sec:intro-complexity-measures}

Most sample complexity bounds in learning theory rely on detailed structures of the hypothesis class $\hypoclass$, and they hold for arbitrary distributions over the data domain. 
In particular, they build on various complexity measures of the hypothesis class, including the covering number \cite{AnthonyB/book/2009}, Vapnik-Chervonenkis (VC) dimension~\cite{VapnikC/2015}, Natarajan dimesnion~\cite{Natarajan/ML/1989}, pseudo-dimension~\cite{Pollard/1990}, fat-shattering dimension~\cite{BartlettLW/JCSS/1996}, Rademacher complexity~\cite{BartlettM/JMLR/2002, KoltchinskiiP/2000}, local Rademacher complexity~\cite{BartlettBM/COLT/2002}, etc.
Informally, each complexity measure provides a parameter $d$ which represents the ``degrees-of-freedom'' of the hypothesis class $\hypoclass$, and the corresponding sample complexity upper bound has the form $\tilde{O} \big( \frac{d}{\epsilon^2} \big)$.

For example, the VC dimension characterizes the sample complexity of binary classification problems, and the Natarajan dimesnion captures that of multiclass classification problems (see, e.g., \citet{Shalev-ShwartzB/book/2014}), \emph{if the underlying distribution could be arbitrary}.

Further, the example of learning revenue-optimal auctions from data, in particular, the special case of selling a single item to $n$ bidders, has been investigated using the covering number (e.g., \citet{DevanurHP/STOC/2016}, \citet{GonczarowskiN/STOC/2017}), pseudo-dimension (e.g., \citet{MorgensternR/NIPS/2015}), and Rademacher complexity (e.g., \citet{Syrgkanis/NIPS/2017}).
The ``degrees-of-freedom'' bounds in these works are all $\tilde{O}\big(\frac{n}{\epsilon}\big)$ and thus, lead to the same $\tilde{O} \big( \frac{n}{\epsilon^3} \big)$ sample complexity upper bound.%
\footnote{Nonetheless, these works are different in that they either prove sample complexity bounds for different families of distributions beyond the $[0, 1]$-bounded ones, and/or provide slightly different bounds in the logarithmic factors.}
The bound once again holds for \emph{arbitrary distributions of the valuations, even correlated ones}, although the problem of revenue-optimal auction design often considers product value distributions.

\subsection{Our Contributions: a Theory of Generalization on Product Distributions}

While the above theories suggest that simpler hypotheses generalize better, it has been increasingly important to consider complex ones.
On the one hand, deep neural networks generalize surprisingly well on classification problems on real-world data despite their complexity (e.g., \citet{ZhangBHRV/ICLR/2017}).
On the other hand, the revenue-optimal auction for selling even two heterogeneous items could have an infinite menu complexity (e.g., \citet{DaskalakisDT/EC/2013}).
To this end, this paper asks:

\begin{quote}
    \emph{Is there a complementary theory of generalization building on the simplicity of data instead of the hypothesis class?}
\end{quote}

In particular, it is standard to assume that the data is drawn from a product distribution in optimization problems in the Bayesian model, including the aforementioned revenue-maximization problem, and prophet inequality and Pandora's problem.
Can we get sample complexity bounds from the independence of data dimensions, and only minimum knowledge about the hypotheses?

\paragraph{Implicit Attempts in Previous Works.}
We first review several recent sample complexity bounds for specific optimization problems that implicitly explore the power of independent data dimensions.
\citet{ColeR/STOC/2014} and \citet{RoughgardenS/EC/2016} used independence in the single-parameter revenue maximization problems to analyze coordinate-wise the convergence of the empirical distribution to the true distribution.
\citet{correa2018prophet} employed a similar approach on prophet inequality.
Their analyses of convergence, however, are problem dependent.

\citet{CaiD/FOCS/2017} proposed a hybrid argument that used the independence of data dimensions in multi-item auctions to derive sample complexity bounds from complexity measures of the hypothesis class w.r.t.\ each coordinate of the data domain. 
Their hybrid approach benefits from independence, yet still relies on complexity measures of the hypothesis class.

\citet{focs/GonczarowskiW18} and \citet{guo2019settling} are the closest to this paper.
\citet{focs/GonczarowskiW18} exploited independence in multi-parameter revenue maximization to construct an improved covering number that holds specifically on product distributions.
Although they did not explicitly ask the above conceptual question, their techniques implicitly showed generalization of complex hypotheses on product distributions;
the resulting bound is inferior in the logarithmic factor compared to Theorem~\ref{thm:error_finite} in this paper.
\citet{guo2019settling} used independent data dimensions and a notion of monotonicity to derive optimal sample complexity bounds for single-parameter revenue maximization problems in the matroid setting.
It is the closest to this paper.
Part of our results can be viewed as generalizing theirs to all problems with the same notion of monotonicity.

\paragraph{Our Results.}
The contributions of the paper are two general sample complexity bounds from the independence of data dimensions, unrelated to any complexity measures of the hypothesis class.
Both results use the same algorithm which we call the \emph{product empirical reward maximizer/risk minimize} (PERM).
It selects the best hypothesis w.r.t.\ a product empirical distribution such that each coordinate is a uniform distribution over the corresponding coordinate of the samples.
This is different from the usual notion of empirical distribution, i.e., the uniform distribution over sample vectors, which is not a product distribution in general.
The first result considers finite data domains.

\begin{informal}
    Suppose the data has $n$ independent dimensions, each of which takes up to $k$ possible values.
    Then, $O \big( \frac{nk}{\epsilon^2} \log \frac{1}{\delta} \big)$ samples are sufficient for learning any hypothesis class up to $\epsilon$-optimal with probability at least $1-\delta$.
    Further, $\Omega \big( \frac{nk}{\epsilon^2} \big)$ samples are necessary.
\end{informal}

The proof of the upper bound is simple in hindsight. 
We bound the total variation distance between the product empirical distribution and the true distribution, using the connection between the total variation distance and the Hellinger distance, and a vector concentration inequality.

Despite its simplicity, the above result gives strong sample complexity upper bounds for the aforementioned optimization problems.
In single-parameter revenue maximization, e.g., single-item auctions, the value can be discretized to multiples of $\epsilon$~\cite{DevanurHP/STOC/2016}; 
replacing $k = \frac{1}{\epsilon}$ gives an $O \big( \frac{n}{\epsilon^3} \log \frac{1}{\delta} \big)$ upper bound.
In other words, without any knowledge of auction theory, other than the fact that the value domain could be discretized, the theorem improves the sample complexity upper bounds derived from various complexity measures based on detailed understandings of single-parameter auctions as discussed in Section~\ref{sec:intro-complexity-measures}, in the log factor.

In multi-parameter revenue maximization, e.g., with one unit-demand bidder and $n$ items, the value domain can be discretized to multiples of $\epsilon^2$~\cite{BalcanBHM/JCSS/2008};
hence, letting $k = \frac{1}{\epsilon^2}$ gives an $O \big(\frac{n}{\epsilon^4} \log \frac{1}{\delta} \big)$ upper bound, improving the state-of-the-art by \citet{focs/GonczarowskiW18} in the log factor.%
\footnote{\citet{focs/GonczarowskiW18} only claim polynomial sample complexity;
the stated bound is derived using their techniques to the best of our efforts.}

Further, we show that the type domain in prophet inequality can be discretized to multiples of $\epsilon$, leading to an $O \big( \frac{n}{\epsilon^3} \log \frac{1}{\delta} \big)$ upper bound.
It improves the best previous bound by \citet{correa2018prophet} in the dependence in $n$, and the concurrent effort by \citet{itcs/RubinsteinW2019} in the dependence in $\epsilon$.
In prophet inequality with i.i.d.\ rewards in particular, it implies that $\tilde{O} \big( n \big)$ samples are sufficient to learn a $0.745$-competitive algorithm, answering an open question by \citet{correa2018prophet}.%
\footnote{The best algorithm, with full knowledge of the distribution, is strictly better than $0.745$-competitive. 
Hence, we may consider $\epsilon$ a constant in this result. 
Further, unlike our model, \citet{correa2018prophet} consider unbounded distributions and multiplicative approximation;
nonetheless, Appendix~\ref{app:prophet-unbounded} shows how to get the stated bounds in their model.}

Finally, we show that the type domain of Pandora's problem can also be discretized to multiples of $\epsilon$, giving the first polynomial sample complexity bound for the problem.

\begin{table}[t]
    \centering
    \renewcommand{\arraystretch}{1.25}
    \begin{tabular}{|c|c|c|c|}
        \hline
        \multirow{2}{*}{}
        & \multicolumn{2}{|c|}{\bf This paper} & \multirow{2}{*}{Previous results} \\
        \cline{2-3}
        & Finite domain (\S\ref{sec:finite}) & Strong monotonicity (\S\ref{sec:strong-monotone}) & \\
        \hline
        General bound & $\bm{O \big( \frac{nk}{\epsilon^2}\,\textbf{log}\,\frac{1}{\delta} \big)}$ & $\bm{\tilde{O} \big( \frac{n}{\epsilon^2} \big)}$ & - \\
        \hline
        Single-parameter & $O \big( \frac{n}{\epsilon^3} \log \frac{1}{\delta} \big)$ & $\bm{\tilde{O} \big( \frac{n}{\epsilon^2} \big)}$ & $\tilde{O} \big( \frac{n}{\epsilon^2} \big)$~\cite{guo2019settling} \\
        \hline
        Multi-parameter & $\bm{O \big( \frac{n}{\epsilon^4}\,\textbf{log}\,\frac{1}{\delta} \big)}$ & - & $O \big( \frac{n}{\epsilon^4} \log \frac{n}{\epsilon \delta} \big)$~\cite{focs/GonczarowskiW18} \\
        \hline
        Prophet inequality & $O \big( \frac{n}{\epsilon^3} \log \frac{1}{\delta} \big)$ & $\bm{\tilde{O} \big( \frac{n}{\epsilon^2} \big)}$ & $\tilde{O} \big( \frac{n^2}{\epsilon^2} \big)$~\cite{correa2018prophet}, $\tilde{O} \big( \frac{n}{\epsilon^6} \big)$~\cite{itcs/RubinsteinW2019} \\
        \hline
        Pandora's problem & $O \big( \frac{n^3}{\epsilon^3} \log \frac{1}{\delta} \big)$ & $\bm{\tilde{O} \big( \frac{n}{\epsilon^2} \big)}$ & - \\
        \hline
    \end{tabular}
    \caption{Summary of sample complexity upper bounds, in comparisons with the state-of-the-art. 
    The results in bold are the best upper bounds in different settings.
    We use single-item auctions, and $n$-item auctions with a unit-demand bidder as the running examples of single- and multi-parameter revenue maximization. 
    The bounds may vary in other settings; see Sec.~\ref{sec:finite} and \ref{sec:strong-monotone}.}
    \label{tab:summary}
\end{table}

Our second result revisits the notion of strong monotonicity, a key ingredient of the optimal sample complexity bounds for single-parameter revenue maximization by \citet{guo2019settling}.
Strong monotonicity means that the expected value of the optimal hypothesis w.r.t.\ a distribution $\vec{\tilde{D}}$ does not decrease when it is applied to another distribution $\vec{D}$ that stochastically dominates $\vec{\tilde{D}}$.
We generalize the analysis by \citet{guo2019settling} to any strongly monotone problem.

\begin{informal}
    Suppose the data has $n$ independent dimensions. 
    Then, $\tilde{\Theta}\big( \frac{n}{\epsilon^2} \big)$ samples are sufficient and necessary for learning an arbitrary strongly monotone hypothesis class.
\end{informal}

Further, we show that prophet inequality and Pandora's problem are both strongly monotone.
Using this theorem, we get an $\tilde{O} \big( \frac{n}{\epsilon^2} \big)$ upper bound for single-parameter revenue maximization, prophet inequality, and Pandora's problem.%
\footnote{Concurrently and independently, \citet{fu2020learning} also proved an $\tilde{O} \big( \frac{n}{\epsilon^2} \big)$ upper bound for Pandora's problem.}
The linear dependence on the data dimension $n$ is tight for all three problems.
We remark that while the bound for single-parameter revenue maximization is the same as that by \citet{guo2019settling}, ours directly uses the PERM, which corresponds to the empirical Myerson auction, while that by \citet{guo2019settling} needs an appropriate regularization to the product empirical distributions and uses the corresponding regularized empirical Myerson auction.

\section{Preliminaries}
\label{sec:preliminary}

\subsection{Model}

A \emph{general learning problem} (e.g., Chapter 1.4 of \citet{Vapnik/2013}) is defined by a \emph{hypothesis class} denoted as $\hypoclass$.
We will abuse notation and refer to the problem defined by a hypothesis class $\hypoclass$ as problem $\hypoclass$.
Each \emph{hypothesis} $\hypo \in \hypoclass$ is a mapping from $\vec{\typespace} = \typespace_1 \times \typespace_2 \times \dots \times \typespace_n$ to $[0, 1]$, where $\typespace_i \subseteq \R$ is the domain of the $i$-th coordinate of the data type.
We will refer to $n$ as the \emph{data dimension} of the problem, to make a distinction with various learning dimensions in the literature which measure the complexity of the hypothesis class.
For a concrete running example, readers may think of $\typespace_1 = \typespace_2 = \dots = \typespace_n = [0, 1]$. 
For any data type $\vec{\type} \in \vec{\typespace}$, and any hypothesis $\hypo \in \hypoclass$, $\hypo(\vec{\type})$ is the reward obtained by hypothesis $\hypo$ on a data point of type $\vec{\type}$.

Given a distribution $\vec{D}$ over $\vec{\typespace}$, we seek to pick a hypothesis $\hypo$ to maximize the expected reward:
\[
    \hypo \big( \vec{D} \big) \defeq \E_{\vec{\type} \sim \vec{D}} \big[ \hypo \big( \vec{\type} \big) \big]
    ~.
\]

Further, let $\opt_{\hypoclass} \big( \vec{D} \big)$ denote the optimal expected reward.
We will omit the subscript $\hypoclass$ for brevity when the hypothesis class is clear from the context.
\[
    \opt_{\hypoclass} \big( \vec{D} \big) \defeq \sup_{\hypo \in \hypoclass} \hypo \big( \vec{D} \big)
    ~.
\]

Throughout this paper, we will assume that $\vec{D} = D_1 \times D_2 \times \dots \times D_n$ is a product distribution, as it is a standard assumption in all examples considered in the paper. 
See Section~\ref{sec:examples} for details.

A \emph{learning algorithm} for a problem $\hypoclass$ takes $N$ i.i.d.\ samples from the underlying distribution $\vec{D}$ as input and returns a hypothesis $\hypo \in \hypoclass$.
Let $E_i$ denote the uniform distribution over the $i$-th coordinate of the samples.
We call $\vec{E} = E_1 \times E_2 \times \dots \times E_n$ the \emph{product empirical distribution}, and its optimal hypothesis $\hypo_{\vec{E}}$ the \emph{product empirical reward maximizer} (PERM).

For any $0 \le \epsilon \le 1$, a hypothesis $\hypo$ is an \emph{$\epsilon$-additive approximation} if:
\[
    \hypo \big( \vec{D} \big) \ge \opt \big( \vec{D} \big) - \epsilon
    ~.
\]

The \emph{sample complexity} of a problem $\hypoclass$ is the minimum number of samples $N$ for which there is a learning algorithm so that, \emph{for any distribution $\vec{D}$}, it takes $N$ i.i.d.\ samples and returns an $\epsilon$-additive approximation with probability at least $1 - \delta$.
The sample complexity bounds in this paper depend on the data dimension $n$, the approximation parameter $\epsilon$, and the confidence parameter $\delta$.
In our first set of results, they further depend on the sizes of the data type domain $\typespace_i$'s.
Importantly, they are independent of any complexity measure of the hypothesis class $\hypoclass$.

\subsection{Examples}
\label{sec:examples}

Next, we define several example problems for which the theory developed in this paper improves or matches the state-of-the-art sample complexity bounds.
We define each problem only with the minimum detail necessary for verifying that it is a special case of the above model.
In particular, we intentionally do not characterize the optimal hypothesis to stress the main feature of our theory:
\emph{it requires almost no knowledge of the hypothesis class}; 
instead, it only needs that (1) $\vec{D}$ is a product distribution, and (2) some generic structural property, e.g., the data type domain can be discretized, or the problem satisfies strong monotonicity, which will be discussed in more details in Section~\ref{sec:strong-monotone}.

\paragraph{Single-parameter Revenue Maximization.}
For simplicity of exposition, we use single-item auctions with $n$ bidders as the running example.
Each bidder $i$ has a type $\type_i \in [0, 1]$ that represents its value for the item, and is drawn independently from $D_i$.
An auction $A$ maps any type profile $\vec{\type}$ to an allocation $\vec{x} \in [0, 1]^n$, $\|\vec{x}\|_1 \le 1$, and a payment vector $\vec{p} \in [0, 1]^n$.
For any bidder $i$, $x_i$ is the probability that the bidder gets the item, and $p_i$ is its payment.
Its utility equals $x_i \type_i - p_i$.

A bidder's type is private information known only to itself;
therefore, the auctioneer must ask the bidders to report the values.
Hence, the literature focuses on dominant-strategy incentive compatible (DSIC) auctions, which ensure that for any bidder, reporting the value truthfully always maximizes its utility.
The goal is to pick a DSIC auction to maximize the expected revenue. 
Readers are referred to \citet{Myerson/MOR/1981} for a characterization of the revenue optimal auction.

To place it in our framework, define the hypothesis class $\hypoclass$ by having a hypothesis $\hypo_A$ for every DSIC auction $A$, such that $\hypo_A(\vec{\type})$ equals the revenue of auction $A$ on a type profile $\vec{\type}$.

Readers who are familiar with auction theory may verify that the techniques in Section~\ref{sec:finite} apply to arbitrary single-parameter problems, and those in Section~\ref{sec:strong-monotone} apply to the matroid setting.%
\footnote{This part of our results rely on a notion called strong revenue monotonicity, which we will discuss in more details in Section~\ref{sec:strong-monotone}. 
It is only known to hold in the matroid setting. 
Whether it further generalizes is an open problem.}

\paragraph{Multi-parameter Revenue Maximization.}
For simplicity of exposition, we use single-bidder auctions with $n$ items as the running example.
Section~\ref{sec:finite} will discuss the extension to multi-bidder multi-item auctions.
The bidder has a type $\vec{\type} \in [0, 1]^n$ such that $\type_i$ is its value for item $i$, and is drawn independently from $D_i$.
There are various settings with different definitions of the bidder's value for subsets of items.
The most-studied ones are the \emph{unit-demand} bidder, whose value for a subset of items is the maximum value for a single item in the subset, and the \emph{additive} bidder, whose value is the sum of item values it gets.
An auction $A$ maps any type $\vec{\type}$ to a subset of items to be allocated to the bidder and its payment.
The bidder's utility is equal to its value for the allocated subset minus the payment.
The goal is to pick a DSIC auction to maximize the expected revenue.
Readers are referred to \citet{CaiDW/FOCS/2012} for an LP-based characterization of the optimal auction.

Similar to the single-parameter setting, define the hypothesis class $\hypoclass$ by having a hypothesis $\hypo_A$ for every DSIC auction $A$, such that $\hypo_A(\vec{\type})$ equals the revenue of auction $A$ on type $\vec{\type}$.

\paragraph{Prophet Inequality.}
Consider $n$ rewards which arrive one at a time; 
each reward $t_i \in [0, 1]$ is drawn independently from $D_i$.
On observing each reward, the algorithm must immediately decide whether to take it or not; 
it can take at most one reward.
The goal is to maximize the expected reward.
Readers are referred to \citet{Samuel-Cahn/AnnProbab/1984} for an algorithm that gets at least a half of the expected max reward, and \citet{CorreaFHOV/EC/2017} for an improved algorithm in the case of i.i.d.\ rewards that gets a $0.745$ fraction of the expected max reward.
Some readers may know it as the optimal stopping problem, while it is often known as prophet inequality in theoretical computer science.

To put it in our framework, define the hypothesis class $\hypoclass$ by having a hypothesis $\hypo_A$ for every algorithm $A$, such that $\hypo_A(\vec{\type})$ equals what the algorithm gets when the reward sequence is $\vec{\type}$.

\paragraph{Pandora's Problem.}
Consider $n$ boxes; 
each box $i$ has a reward $\type_i \in [0, 1]$ drawn from $D_i$ and a fixed cost $c_i \in [0, 1]$ for opening it.
In each round, the algorithm decides if to take the best observed reward, or to open a new box.
The goal is to maximize the reward it gets minus the total cost.

To interpret it in the our framework, define the hypothesis class $\hypoclass$ by having a hypothesis $\hypo_A$ for every algorithm $A$. 
If we let $\hypo_A(\vec{\type})$ equals what the algorithm gets minus the cost when the reward sequence is $\vec{\type}$, its range would be $[-n, 1]$ instead of $[0, 1]$.
In the main text, we use a simple normalization, which let $\hypo_A$ be the above value plus $n$ and scaled by $\frac{1}{n+1}$.
Appendix~\ref{app:pandora-improved} presents a more specialized method that gives the tight sample complexity bound. 

\subsection{Metrics for Probability Distributions}

Consider two distributions $P, Q$ over a sample domain $\typespace$.
For concreteness, think of $\typespace$ as a cube in the Euclidean space, e.g., $[0, 1]$ or $[0, 1]^n$.

\paragraph{Total Variation Distance.}
The \emph{total variation distance} is a half of the $L_1$ distance:
\begin{equation}
    \label{eqn:total-variation}
    \TV \big( P, Q \big) = \frac{1}{2} \| P - Q \|_1
    ~.
\end{equation}

The following useful fact about total variation distance follows by its definition.

\begin{lemma}
    \label{lem:total-variation}
    For any distributions $P, Q$ over a sample domain $\typespace$, and any function $h : \typespace \mapsto [0, 1]$:
    \[
        \big| h \big( P \big) - h \big( Q \big) \big| \le \TV \big( P, Q \big)
        ~.
    \]
\end{lemma}

Recall that we are interested in multi-dimensional product distributions.
It is hard to directly measure the total variation distance among such distributions.
The standard method is to instead consider either the Kullback-Leibler divergence or the Hellinger distance; 
they are both additive in that we can account for the distance in each coordinate separately, and both can be related to the total variation distance.
This paper uses the latter because it has better properties.

\paragraph{Hellinger Distance.}
The \emph{Hellinger distance} between $P$ and $Q$, denoted as $\Hellinger \big( P, Q \big)$, is given by:
\[
    \Hellinger^2 \big( P, Q \big) = \frac{1}{2} \int_{\typespace} \big( \sqrt{dP} - \sqrt{dQ} \big)^2
    ~.
\]

More formally, for any measure $\lambda$ over $\typespace$ so that both $P$ and $Q$ are absolutely continuous w.r.t.\ $\lambda$, let $\frac{dP}{d\lambda}$ and $\frac{dQ}{d\lambda}$ be the Radon-Nikodym derivatives. 
We have:
\[
    \Hellinger^2 \big( P, Q \big) = \frac{1}{2} \int_{\typespace} \bigg( \sqrt{\frac{dP}{d\lambda}} - \sqrt{\frac{dQ}{d\lambda}} \bigg)^2 d\lambda
    ~.
\]

For example, if $P$ and $Q$ are continuous over $[0, 1]$ with density functions $p$ and $q$, or if $P$ and $Q$ are distributions over a discrete set $\typespace$ with probability mass functions $p$ and $q$, we have:
\[
    \Hellinger^2 \big( P, Q \big) = \frac{1}{2} \int_0^1 \big( \sqrt{p(\type)} - \sqrt{q(\type)} \big)^2 d\type
    \quad \textrm{or} \quad
    \Hellinger^2 \big( P, Q \big) = \frac{1}{2} \sum_{\type \in \typespace} \big( \sqrt{p(\type)} - \sqrt{q(\type)} \big)^2
    ~.
\]

The next two lemmas relate the Hellinger distance with the total variation distance, and formalize its additivity.
Readers are referred to \citet{GibbsS/ISR/2002} for details of these properties and a comprehensive discussion on different metrics for probability distributions.

\begin{lemma} 
    \label{lem:hellinger-and-total-variation}
    Suppose $P$ and $Q$ are distributions over a sample domain $\typespace$.
    Then, we have:
    \[
        \Hellinger^2 \big( P, Q \big) \le \TV \big( P, Q \big) \le \sqrt{2} \Hellinger \big( P, Q \big)
        ~.
    \]
\end{lemma}

\begin{lemma}
    \label{lem:hellinger-decompose}
    Suppose $\vec{P}$ and $\vec{Q}$ are product distributions over $\vec{\typespace} = \typespace_1 \times \typespace_2 \times \dots \times \typespace_n$.
    Then:
    \[
        1 - \Hellinger^2 \big( \vec{P}, \vec{Q} \big) = \prod_{i = 1}^n \big( 1 - \Hellinger^2 \big( P_i, Q_i \big) \big)
        ~.
    \]
\end{lemma}

\subsection{Vector Concentration Inequality}

%
	%
	%

We will use the following Bernstein-style concentration inequality that bounds the $\ell_2$-norm of the sum of independent random vectors. 
\begin{lemma}[Equation 6.12 of \cite{ProbBanach}]\label{lem:vector-bernstein}
	Let $\vec{X_1}, \vec{X_2}, \dots, \vec{X_N}$ be i.i.d.\ random vectors in $\mathbb{R}^d$ such that $\E \big[ \Vert \vec{X_i} \Vert_2^2 \big] \le \sigma^2$, and $ \left\Vert \vec{X_i} \right\Vert_2 \le M$ for some constant $M > 0$.
	Then, for any positive $\Delta$:
	\[
		\textstyle
		\Pr \left[ ~ \left| \big\Vert \sum_{i=1}^N \vec{X_i} \big\Vert_2 - \E \big[ \big\Vert \sum_{i=1}^N \vec{X_i} \big\Vert_2 \big] \right| > \Delta ~ \right] \le 2 \exp \Big( -\frac{\Delta^2}{ 2N \sigma^2} \big( 2 - \exp ( \frac{2M\Delta}{N\sigma^2} ) \big) \Big)
		~.
	\]
\end{lemma}

\section{Problems with Finite (Discretized) Domain}
\label{sec:finite}

\subsection{General Sample Complexity Bounds}

This section offers a theory of generalization for arbitrary hypotheses with finite domain. 
    
\begin{theorem} 
    \label{thm:error_finite}
    For any distribution $\mathbf{D}$ on $\vec{\typespace}$ such that $|\typespace_i| \le k$ for all $1 \le i \le n$, suppose for some sufficiently large constant $C > 0$, the number of samples is at least:
    \[
        C \cdot \frac{nk}{\epsilon^2} \log \frac{1}{\delta}
    \]
    Then, with probability at least $1-\delta$, for any $\hypo : \vec{\typespace} \to [0,1]$, we have:
    \[
        \big| \hypo(\vec{D}) - \hypo(\vec{E}) \big| \le \epsilon
        ~.
    \]
    In particular, the PERM is an $\epsilon$-additive approximation.
\end{theorem}
    

The proof is simple in hindsight.
The key is bounding the convergence of the product empirical distribution to the true distribution in term of the Hellinger distance, as in the next lemma.

\begin{lemma}\label{lem:hellinger_distance}
    With probability at least $1 - \delta$:
    %
    \[
        \hel(\mathbf{D},\mathbf{E}) = O \bigg( \frac{nk}{N} \log \frac{1}{\delta} \bigg) 
        ~.
    \]
\end{lemma}

%
%

%
%

We proceed with the proof of Theorem~\ref{thm:error_finite} assuming the correctness of the lemma, whose proof is deferred to the end of the section.

\begin{proof}[Proof of Theorem~\ref{thm:error_finite}]
    By Lemma~\ref{lem:hellinger_distance} and the stated number of samples, with probability at least $1-\delta$:
    \[
        H^2(\mathbf{D}, \mathbf{E}) \le \frac{\epsilon^2}{2}
        ~.
    \]
    
    Further by Lemma~\ref{lem:hellinger-and-total-variation}, we have:
    \[
        \delta(\vec{D}, \vec{E}) \le \sqrt{2} \cdot \Hellinger (\vec{D}, \vec{E}) \le \epsilon
        ~.
    \]
    
    Then, the theorem follows by Lemma~\ref{lem:total-variation}.
\end{proof}


We complement Theorem~\ref{thm:error_finite} with a matching lower bound up to a logarithmic factor.
The proof is deferred to Appendix~\ref{app:finite-lb}.

\begin{theorem}
    \label{thm:finite_domain_lower_bound}
    There is a problem $\hypoclass$ on a finite domain $\vec{\typespace}$ with $|\typespace_i|=k$ for all $1 \le i \le n$, such that no algorithm gives an expected $\epsilon$-additive approximation if the number of samples is less than:
    \[
        c \cdot \frac{nk}{\epsilon^2} 
        ~,
    \]
    for a sufficiently small constant $c > 0$.
\end{theorem}

\subsection{Applications}
Although some problems are defined on continuous domains, most can be discretized, including all examples defined in Section~\ref{sec:examples}.
Observe that by rounding each sample to the closest discretized value, we effectively sample from the discretized distribution.
Next we discuss the applications of Theorem~\ref{thm:error_finite} on the examples;
the discretization arguments are either from previous works, or deferred to the appendix because they do not provide much insight. 
While some results will be subsumed by those in the next section, they are already comparable with the state-of-the-art in meaningful ways.
We restress that our bounds are obtained knowing effectively nothing about the problems other than that the domains can be discretized, while previous works generally rely on detailed problem structures. 
We believe the same approach can be applied to other problems not covered in this paper.

\paragraph{Single-parameter Revenue Maximization.}
%
%
%
\citet{DevanurHP/STOC/2016} showed that we may w.l.o.g.\ round values down to multiples of $\epsilon$ in single-parameter revenue maximization if the target is an $O(\epsilon)$-additive approximation. 
Hence, with $k = \frac{1}{\epsilon}$, Theorem~\ref{thm:error_finite} matches the previous results by \citet{MorgensternR/NIPS/2015}, \citet{DevanurHP/STOC/2016}, and \citet{Syrgkanis/NIPS/2017} discussed in Section~\ref{sec:intro-complexity-measures}, i.e., the best bounds before the recent work of \citet{guo2019settling}.
In fact, our bound is better in the logarithmic factors.

\begin{theorem}
    In a single-item auction with $n$ bidders whose values are bounded in $[0, 1]$, the sample complexity is at most $O \big(\frac{n}{\epsilon^3} \log \frac{1}{\delta} \big)$.
\end{theorem}

\paragraph{Multi-parameter Revenue Maximization.} We consider the case for selling $n$-items to a unit-demand (resp., additive) bidder. 
It is known that rounding the item values down to multiples of $\epsilon^2$ (resp., $\epsilon^2/n$) for a unit-demand (resp. additive) bidder is w.l.o.g.\ due to a reduction from approximate DSIC auctions to DSIC auctions, which states any $\epsilon^2$-DSIC mechanism can be transformed into a DSIC mechanism with at most $\epsilon$ loss in revenue (see, e.g., \citet{BalcanBHM/JCSS/2008}, attributed to Nisan).
Therefore, Theorem~\ref{thm:error_finite} gives the following bounds that improve the best known results by \citet{focs/GonczarowskiW18} in the log factor.%

\begin{theorem}
    In a multi-item auction with $n$ items and a unit-demand bidder whose values are bounded in $[0, 1]$, the sample complexity is at most $O \big( \frac{n}{\epsilon^4} \log \frac{1}{\delta} \big)$.
\end{theorem}

\begin{theorem}
    In a multi-item auction with $n$ items and an additive bidder whose values are bounded in $[0, 1]$, the sample complexity is at most $O \big( \frac{n^2}{\epsilon^4} \log \frac{1}{\delta} \big)$.%
    \footnote{In the additive case, the optimal revenue is bounded by $[0,n]$. The stated bound is an $\epsilon$-approximation w.r.t.\ the normalized revenue divided by a factor $\frac{1}{n}$. The bound would be $\tilde{O} \big( \frac{n^4}{\epsilon^4} \big)$ for an $\epsilon$-approximation without normalization.}
\end{theorem}

\paragraph{Multiple Bidders.} 
For $n$-bidder $m$-item auctions, Theorem~\ref{thm:error_finite} gives an $\tilde{O}\big( \frac{mnk}{\epsilon^2} \big)$ bound if the buyers' value domains are finite with size $k$, improving those by \citet{focs/GonczarowskiW18} in the log factor.
For continuous value domains, however, there is no existing transformation from approximate DSIC to DSIC auctions in this more general setting;
as a result, the aforementioned discretization no longer works.
Either we settle with approximate DSIC auctions, or need to know more about the relation between approximate and exact DSIC auctions, which is a fundamental question on its own in auction theory.
Readers are referred to \citet{focs/GonczarowskiW18} for an extensive discussion on this topic.

\paragraph{Prophet Inequality.}
We show that the rewards in prophet inequality can be discretized w.l.o.g.\ to multiples of $\epsilon$.
Hence, letting $k = \frac{1}{\epsilon}$, we get a sample complexity bound that improves the recent work of \citet{correa2018prophet} in the dependence in $n$, and the concurrent result by \citet{itcs/RubinsteinW2019} in the dependence in $\epsilon$, with bounded rewards in $[0, 1]$ and additive approximation. 
The discretization argument and an extension to the original setting of \citet{correa2018prophet} with unbounded rewards and multiplicative approximation is deferred to Appendix~\ref{app:app-prophet}.

\begin{theorem}\label{thm:finite-prophet}
    In the prophet inequality setting with $n$ items whose rewards are bounded in $[0,1]$, the sample complexity is at most $O \big(\frac{n}{\epsilon^3} \log \frac{1}{\delta} \big)$.
\end{theorem}

\paragraph{Pandora's Problem.}
Recall that the main text uses a simple normalization by a factor $\frac{1}{n+1}$ to ensure that the hypotheses in Pandora's problem has range $[0, 1]$.
Hence, to get an $\epsilon$-approximation w.r.t.\ Pandora's problem, we need a $\frac{\epsilon}{n+1}$-approximately optimal hypothesis.
Further, we show that the rewards can be w.l.o.g.\ discretized to multiples of $\epsilon$.
Putting together, we get the first polynomial sample complexity for Pandora's problem.
The discretization argument and a more specialized method that gives the optimal sample complexity are deferred to Appendix~\ref{app:pandora}.

\begin{theorem}\label{thm:finite-pandora}
    In the Pandora's problem with $n$ boxes whose rewards and costs are bounded in $[0,1]$, the sample complexity is at most $O \big(\frac{n^3}{\epsilon^3} \log \frac{1}{\delta} \big)$.
\end{theorem}

\subsection{Convergence of Product Empirical Distribution: Proof of Lemma~\ref{lem:hellinger_distance}}
\label{sec:proof-of-hellinger}

We will reduce the problem to a vector concentration inequality.
First:
%
\begin{align*}
    1-\Hellinger^2(\vec{D},\vec{E})
    & =
    \prod_{i=1}^n \big( 1-\Hellinger^2(D_i,E_i) \big) 
    && \text{(Lemma~\ref{lem:hellinger-decompose})} \\
    & \ge 
    1 - \sum_{i=1}^n \Hellinger^2(D_i, E_i)
    ~.
\end{align*}

Hence, it suffices to show that with probability at least $1 - \delta$:
%
\begin{equation}
    \label{eqn:hellinger-per-coorindate}
    \sum_{i=1}^n\Hellinger^2(D_i,E_i) \le O \bigg( \frac{nk}{N} \log \frac{1}{\delta} \bigg) 
    ~.
\end{equation}

By definition, for any $1 \le i \le n$:
\begin{align*}
    \Hellinger^2(D_i,E_i) &= \frac{1}{2} \sum_{t\in T_i} \Big( \sqrt{f_{D_i}(t)} - \sqrt{f_{E_i}(t)} \Big)^2 = \frac{1}{2} \sum_{t\in T_i} \Bigg(\frac{f_{D_i}(t) - f_{E_i}(t)}{\sqrt{f_{D_i}(t)} 
    +\sqrt{f_{E_i}(t)} } \Bigg)^2.
\end{align*}

Next, bound the right-hand-side with the following inequality, which can be viewed a smoothed variant of the connection between the Hellinger distance and $\chi$-square distance.

\begin{lemma}
    \label{lem:reduce-to-chi-square}
    For any $f_D, f_E \ge 0$:
    \[
        \bigg( \frac{f_D - f_E}{\sqrt{f_D} + \sqrt{f_E}} \bigg)^2 \le \frac{\big(f_D - f_E \big)^2}{\max \big\{ f_D, \frac{1}{N} \log \frac{1}{\delta} \big\}} + \frac{1}{N}\log \frac{1}{\delta}
        ~.
    \]
\end{lemma}

\begin{proof}
    If $f_E > \frac{1}{N} \log \frac{1}{\delta}$ or $f_D >\frac{1}{N} \log \frac{1}{\delta}$, the left-hand-side is at most the first term on the right-hand-side.
    Otherwise, the left-hand-side is at most $\max \{ f_D, f_E \} \le \frac{1}{N}\log (\frac{1}{\delta})$.
\end{proof}

Sum Eqn.~\eqref{eqn:hellinger-per-coorindate} over $1 \le i \le n$, and apply Lemma~\ref{lem:reduce-to-chi-square} to the right-hand-side:
\[
    \sum_{i=1}^n \Hellinger^2 (D_i, E_i) \le \frac{1}{2} \sum_{i=1}^n \sum_{t\in T_i} \frac{\big(f_{D_i}(t) - f_{E_i}(t)\big)^2}{\max \big\{ f_{D_i}(t), \frac{1}{N} \log \frac{1}{\delta} \big\}} + \frac{nk}{N} \log \frac{1}{\delta}
    ~.
\]
 
Therefore, it suffices to show that with probability at least $1 - \delta$:
\begin{equation}
    \label{eqn:finite-domain-convergence-coordinatewise}
    \sum_{i=1}^n \sum_{t \in T_i} \frac{\big(f_{D_i}(t) - f_{E_i}(t)\big)^2}{\max \big\{ f_{D_i}(t), \frac{1}{N} \log \frac{1}{\delta} \big\}} \le O\left( \frac{nk}{N} \log \frac{1}{\delta} \right)
    ~.
\end{equation}

To interpret the left-hand-side as the squared $\ell_2$-norm of the sum of i.i.d.\ vectors, we associate each sample $\vec{s_j} \sim \mathbf{D}$, $1 \le j \le N$, with a $\sum_{i=1}^n|T_i|$-dimensional random vector $\vec{X_j}$.
Concretely, for any $j\in[N], i\in [n]$ and $t\in T_i$:
\[
    X_{jit} = \frac{\mathbf{1}[s_{ji} = t]-f_{D_i}(t)}{\sqrt{\max\big\{f_{D_i}(t),\frac{1}{N}\log \frac{1}{\delta}\big\}}} 
    ~.
\]

Then, Eqn.~\eqref{eqn:finite-domain-convergence-coordinatewise} can be restated as:
\[
    \Big\Vert~\frac{1}{N} \sum_{j=1}^N \vec{X_j}~\Big\Vert_2^2 \le O \left( \frac{nk}{N} \log \frac{1}{\delta} \right)
    \quad\text{, or equivalently}\qquad
    \Big\Vert~\sum_{j=1}^N \vec{X_j}~\Big\Vert_2^2 \le O \left( Nnk \log \frac{1}{\delta} \right)
    ~.
\]

We establish the following properties of the random vectors.

\begin{lemma}
    \label{lem:vector-property}
    The i.i.d.\ random vectors $\vec{X_j}$, $1 \le j \le N$, satisfy
    (1) $\E \big[ \vec{X_j} \big] = \vec{0}$,
    (2) $\E \big[ \Vert \vec{X_j} \Vert_2^2 \big] \le nk$, and
    (3) $\Vert \vec{X_j} \Vert_2 \le \sqrt{\frac{N n k}{\log \frac{1}{\delta}}}$.
\end{lemma}

\begin{proof}
    The first property follows by definition.
    The second one is true because:
    \[
        \E \big[ \Vert \vec{X_j} \Vert_2^2 \big] \le \sum_{i=1}^n \sum_{t\in T_i}\frac{\E\big[(\mathbf{1}[\vec{s_j}_i=t]-f_{D_i}(t))^2\big]}{f_{D_i}(t)} = \sum_{i=1}^n \sum_{t\in T_i}(1-f_{D_i}(t))\le nk
        ~.
    \]
    
    The last property follows by $X_{ijk} \le N/\log \frac{1}{\delta}$.
    This is why we need the smoothed variant in Lemma~\ref{lem:reduce-to-chi-square} instead of the original inequality between the Hellinger and $\chi$-square distances.
\end{proof}

By Lemma~\ref{lem:vector-property}, and Lemma \ref{lem:vector-bernstein} with $\sigma^2 = 128nk$,%
\footnote{We need the additional constant $128$ in $\sigma^2$ because the right-hand-side of Lemma~\ref{lem:vector-bernstein} is \emph{not} monotone in $\sigma^2$.}
$M = \sqrt{Nnk/\log \frac{1}{\delta}}$, and $\Delta = 16 \sqrt{N n k \log \frac{1}{\delta}}$, we get that with probability at least $1 - \delta$:
\[
    \bigg|~\Big\Vert \sum_{j=1}^N \vec{X_j} \Big\Vert_2 - \E \Big[ \Big\Vert \sum_{j=1}^N \vec{X_j} \Big\Vert_2 \Big]~\bigg| \le O \bigg( \sqrt{Nnk\log \frac{1}{\delta}} \bigg)
    ~.
\]

Finally, it remains to bound the expected $\ell_2$-norm of $\sum_{j=1}^N \vec{X_j}$:
\begin{align*}
    \E \Big[ \Big\Vert \sum_{j=1}^N \vec{X_j} \Big\Vert_2 \Big]^2
    &
    \le
    \E \Big[ \Big\Vert \sum_{j \in [N]} \vec{X_j} \Big\Vert_2^2 \Big]
    && \text{(Cauchy-Schwarz)} \\
    &
    =
    \sum_{j=1}^N \E \Big[ \big\Vert \vec{X_j} \big\Vert_2^2 \Big]
    && \text{(independence of $\vec{X_j}$'s, and $\E \big[ \vec{X_j} \big] = \vec{0}$ by Lemma~\ref{lem:vector-property})} \\[2ex]
    &
    \le Nnk
    ~.
    && \text{($\E \big[ \Vert \vec{X_j} \Vert_2^2 \big] \le nk$ by Lemma~\ref{lem:vector-property})}
\end{align*}

\section{Strongly Monotone Problems}
\label{sec:strong-monotone}

This section considers the sample complexity of a subset of problems which satisfy a structural property called strong monotonicity, without any restrictions on the supports of the data domain.

\begin{definition}[Strong Monotonicity]
\label{def:strong-monotone}
A problem $\hypoclass$ is \emph{strongly monotone} if 
for any $\vec{D}$, any $\vec{\tilde{D}}$ that is stochastically dominated by $\vec{D}$, and the optimal hypothesis $\hypo_{\vec{\tilde{D}}}$ of $\vec{\tilde{D}}$:
\[
	\hypo_{\vec{\tilde{D}}} \big( \vec{D} \big) \ge \hypo_{\vec{\tilde{D}}} \big( \vec{\tilde{D}} \big)
	= \opt \big( \vec{\tilde{D}} \big)
	~.
\]
\end{definition}

The name is inherited from the context of single-parameter revenue maximization, where each hypothesis $\hypo$ is a DSIC auction, $\vec{v} \sim \vec{D}$ is the value profile, and $\hypo(\vec{D})$ is the expected revenue of the auction over the random realization of a value profile drawn from $\vec{D}$.
Then, the above inequality states that running the optimal auction w.r.t.\ $\vec{\tilde{D}}$, a.k.a., Myerson's auction, on a distribution $\vec{D}$ that stochastically dominates $\vec{\tilde{D}}$, gets at least the optimal revenue w.r.t.\ the dominated distribution $\vec{\tilde{D}}$.
This is precisely the notion of \emph{strong revenue monotonicity} introduced by \citet{DevanurHP/STOC/2016}.
The naming is to make a distinction with the existing weaker notion of revenue monotonicity, which only requires that optimal revenue w.r.t.\ $\vec{D}$ to be weakly larger than that w.r.t.\ $\vec{\tilde{D}}$.
We restate below the weaker notion in the more general context in this paper.

\begin{definition}
    \label{def:weak-monotone}
    A problem $\hypoclass$ is \emph{weakly monotone} if for any $\vec{D}$, and any $\vec{\tilde{D}}$ that is stochastically dominated by $\vec{D}$:
    \[
        \opt \big( \vec{D} \big) \ge \opt \big( \vec{\tilde{D}} \big)
    	~.
    \]
\end{definition}

Finally, we remark that there is an even stronger notion of monotonicity which we call hypothesis-wise monotonicity.

\begin{definition}
    \label{def:hypo-wise-monotone}
    A problem $\hypoclass$ is \emph{hypothesis-wise monotone} if for any $\vec{D}\in \R^n$, any $\vec{\tilde{D}} \in \R^n$ that is stochastically dominated by $\vec{D}$, and any hypothesis $\hypo \in \hypoclass$:
\[
    \hypo \big( \vec{D} \big) \ge \hypo \big( \vec{\tilde{D}} \big)
	~.
\]
\end{definition}

Clearly, hypothesis-wise monotonicity implies strong monotonicity, which in turns implies weak monotonicity.
Weak monotonicity is insufficient for deriving the improved sample complexity bound with the techniques in this section.
Hypothesis-wise monotonicity is too restrictive on the other hand;
in fact, it fails to hold on any example considered in this paper. 

The rest of the section argues that (1) strong monotonicity leads to an improved sample complexity bound, and (2) strong monotonicity holds in all but one examples considered in this paper, improving or matching the state-of-the-art sample complexity bounds. 

\subsection{Sample Complexity Bounds for Strongly Monotone Problems}
\label{sec:strong-monotone-bounds}

The main result for strongly monotone problems is the following improved sample complexity upper bound.
In particular, the bound is independent of the support size of the distributions and, in fact, applies to continuous distributions.

\begin{theorem}
    \label{thm:strong-monotone-upper-bound}
    For any strongly monotone problem $\hypoclass$, suppose the number of samples is at least:
    \[
        C \cdot \frac{n}{\epsilon^2} \log \left( \frac{n}{\epsilon} \right) \log \left( \frac{n}{\epsilon \delta} \right)
        ~,
    \]
    where $C > 0$ is a sufficiently large constant independent of the problem $\hypoclass$. 
    Then, the PERM is an $\epsilon$-additive approximation with probability at least $1 - \delta$.
    
    \medskip
    
    \noindent
    \textbf{Remark:~}
    If the distributions are i.i.d., i.e., $D_i = D^*$ for any $i \in [n]$, it suffices to have the above number of sample values from $D^*$ (rather than vectors from $\vec{D}$) and construct an i.i.d.\ empirical distribution $\vec{E}$ such that each coordinate is a uniform distribution over the samples.
\end{theorem}

The above upper bound is identical to that in the special case of single-parameter revenue maximization by \citet{guo2019settling};
the proof is also similar.
The contributions of this paper are two-folds.
First, we generalize it to arbitrary strongly monotone problems so that it can be further applied to a broader scope of problems, including the prophet inequalities and the Pandora's problem considered in this paper. 
Second, we show that the empirical maximizer itself achieves the optimal sample complexity bound when the reward function is bounded in $[0, 1]$;
in contrast, \citet{guo2019settling} need a regularized version of the empirical distributions called the dominated empirical distributions, and uses the corresponding regularized maximizer.

Before presenting the proof of Theorem~\ref{thm:strong-monotone-upper-bound}, we remark that the above upper bound is tight up to a poly-logarithmic factor, due to an existing lower bound in the special case of single-parameter revenue maximization.

\begin{theorem}
    \label{thm:strong-monotone-lower-bound}
    There is a strongly monotone problem $\hypoclass$ so that if the number of samples is less than:
    \[
        c \cdot \frac{n}{\epsilon^2} 
        ~,
    \]
    where $c > 0$ is a sufficiently small constant, no algorithm gets an expected $\epsilon$-additive approximation.
\end{theorem}

\begin{proof}
    Let $\hypoclass$ be the set of DSIC single-item auctions with $n$ bidders. 
    Restrict the bidders' valuations to be bounded in $[0, 1]$ so that the value/revenue of any hypothesis/auction $\hypo \in \hypoclass$ on any value profile is bounded in $[0, 1]$.
    By \citet{DevanurHP/STOC/2016}, the single-item revenue maximization problem is strongly monotone.
    Further by \citet{guo2019settling}, the sample complexity of $[0, 1]$-bounded valuations and $\epsilon$-additive approximation is at least $\Omega(\tfrac{n}{\epsilon^2})$.
\end{proof}

\subsection{Proof of Theorem~\ref{thm:strong-monotone-upper-bound}}
\label{sec:strong-monotone-upper-bound-proof}


By the Bernstein inequality and union bound, we can relate the CDFs of underlying distribution $\vec{D}$ and the empirical distribution $\vec{E}$ as follows.

\begin{lemma}[e.g., Lemma 5 of \citet{guo2019settling}]
    \label{lem:strong-monotone-D-and-E}
    With probability at least $1 - \delta$, we have that for any $i \in [n]$, and any $\type_i \in [0, 1]$:
    \[
        \big| F_{D_i} (\type_i) - F_{E_i} (\type_i) \big| \le
        \sqrt{\frac{2 F_{D_i}(\type_i) \big( 1-F_{D_i}(\type_i) \big) \ln(2Nn\delta^{-1})}{N} } + \frac{\ln(2Nn\delta^{-1})}{N}
        ~.
    \]
\end{lemma}

The rest of the subsection shows the stated additive approximation factor under the assumption that the inequality in Lemma~\ref{lem:strong-monotone-D-and-E} holds.

We introduce two auxiliary distribution $\vec{\hat{D}}$ and $\vec{\check{D}}$, where the former serves as an upper bound of $\vec{E}$ and the latter as a lower bound.
Concretely, for any $i \in [n]$, define the CDF of $\hat{D}_i$ as follows:
\begin{equation}
    \label{eqn:strong-monotone-Dhat}
    F_{\hat{D}_i} \big( \type_i \big) 
    = \begin{cases}
        1 & x_i = 1 ~; \\
        \max \bigg\{ 0, F_{D_i} (\type_i) - \sqrt{\frac{2 F_{D_i}(\type_i) ( 1-F_{D_i}(\type_i) ) \ln(2Nn\delta^{-1})}{N} } - \frac{\ln(2Nn\delta^{-1})}{N} \bigg\}
        & 0 \le x_i < 1 ~.
    \end{cases} 
\end{equation}

The case of $x_i = 1$ is defined separately because its CDF must be $1$ for any distribution with support bounded in $[0, 1]$.
The other cases are defined to be the smallest possible value of $F_{E_i}\big(x_i\big)$ according to Lemma~\ref{lem:strong-monotone-D-and-E} and the trivial lower bound of $F_{E_i}\big(x_i\big) \ge 0$.

Similarly, for any $i \in [n]$, define the CDF of $\check{D}_i$ as follows:
\begin{equation}
    \label{eqn:strong-monotone-Dcheck}
    F_{\check{D}_i} \big( \type_i \big)
    = \begin{cases}
        0 & x_i = 0 ~; \\
        \min \bigg\{ 1, F_{D_i} (\type_i) + \sqrt{\frac{2 F_{D_i}(\type_i) ( 1-F_{D_i}(\type_i) ) \ln(2Nn\delta^{-1})}{N} } + \frac{\ln(2Nn\delta^{-1})}{N} \bigg\}
        & 0 < x_i \le 1 ~.
    \end{cases}
\end{equation}

Then, the empirical distribution is sandwiched between the auxiliary distributions by definition.

\begin{lemma}
    \label{lem:strong-monotone-auxiliary-and-E}
    Assuming the inequality in Lemma~\ref{lem:strong-monotone-D-and-E}, we have:
    \[
        \vec{\hat{D}} \succeq \vec{E} \succeq \vec{\check{D}}
        ~.
    \]
\end{lemma}

Therefore, we can lower bound the performance of the empirical maximizer, i.e., $\hypo_{\vec{E}}$, on the underlying distribution $\vec{D}$ through a sequence of inequalities below:
\begin{align*}
    \hypo_{\vec{E}} \big( \vec{D} \big)
    &
    \ge \hypo_{\vec{E}} \big( \vec{\hat{D}} \big) - \delta \big( \vec{\hat{D}}, \vec{D} \big)
    && \text{($h_{\vec{E}}$ bounded in $[0, 1]$, Lemma~\ref{lem:total-variation})} \\
    & 
    \ge \hypo_{\vec{E}} \big( \vec{E} \big) - \delta \big( \vec{\hat{D}}, \vec{D} \big) 
    && \text{(strong monotonicity, $\vec{\hat{D}} \succeq \vec{E}$)} \\
    & 
    = \opt \big( \vec{E} \big) - \delta \big( \vec{\hat{D}}, \vec{D} \big) 
    && \text{(definition of $\opt\big(\vec{E}\big)$)} \\
    &
    \ge \opt \big( \vec{\check{D}} \big) - \delta \big( \vec{\hat{D}}, \vec{D} \big) 
    && \text{(weak monotonicity, $\vec{E} \succeq \vec{\check{D}}$)} \\
    & 
    \ge h_{\vec{D}} \big( \vec{\check{D}} \big) - \delta \big( \vec{\hat{D}}, \vec{D} \big) 
    && \textrm{(definition of $\opt\big(\vec{\check{D}}\big)$)} \\
    & 
    \ge h_{\vec{D}} \big( \vec{D} \big) - \delta \big( \vec{\check{D}}, \vec{D} \big) - \delta \big( \vec{\hat{D}}, \vec{D} \big)
    && \text{($h_{\vec{D}}$ bounded in $[0, 1]$, Lemma~\ref{lem:total-variation})} \\
    &
    = \opt \big( \vec{D} \big) - \delta \big( \vec{\check{D}}, \vec{D} \big) - \delta \big( \vec{\hat{D}}, \vec{D} \big) ~.
    &&
    \textrm{(definition of $\opt\big(\vec{D}\big)$)}
\end{align*}

Therefore, it remains to show that with the number of samples stated in the theorem:
\begin{equation}
    \label{eqn:strong-monotone-tv-bound}
    \delta \big( \vec{\check{D}}, \vec{D} \big) \le \frac{\epsilon}{2}
    \quad,\quad 
    \delta \big( \vec{\hat{D}}, \vec{D} \big) \le \frac{\epsilon}{2}
    ~.
\end{equation}

By Lemma~\ref{lem:hellinger-and-total-variation}, it suffices to upper bound the Hellinger distances, as in the following lemmas.

\begin{lemma}
    \label{lem:strong-monotone-Dhat-and-D}
    For any distribution $\vec{D}$ and the corresponding $\vec{\hat{D}}$ defined in Eqn.~\eqref{eqn:strong-monotone-Dhat}, we have:
    \[
        \Hellinger^2 \big( \vec{\hat{D}}, \vec{D} \big) \le O \left( \frac{n}{N} \log \left( \frac{N n}{\delta} \right) \log \left( \frac{N}{\log (Nn\delta^{-1})} \right) \right)
        ~.
    \]
\end{lemma}

\begin{lemma}
    \label{lem:strong-monotone-Dcheck-and-D}
    For any distribution $\vec{D}$ and the corresponding $\vec{\check{D}}$ defined in Eqn.~\eqref{eqn:strong-monotone-Dcheck}, we have:
    \[
        \Hellinger^2 \big( \vec{\check{D}}, \vec{D} \big) \le O \left( \frac{n}{N} \log \left( \frac{N n}{\delta} \right) \log \left( \frac{N}{\log (Nn\delta^{-1})} \right) \right)
        ~.
    \]
\end{lemma}

The proofs of the above lemmas, which we include at the end of the subsection for completeness, are analogous to the proof of a similar lemma w.r.t.\ Kullback-Leibler divergence by \citet{guo2019settling}.
The main difference is that the lemma by \citet{guo2019settling} requires additional conditions that lower bound the probability masses of the two endpoints of the support, while ours do not.

As corollaries of the lemma, with a number of samples 
stated in the lemma, we have:
\[
    \Hellinger^2 \big( \vec{\hat{D}}, \vec{D} \big) \le \frac{\epsilon^2}{8} 
    \quad,\quad
    \Hellinger^2 \big( \vec{\check{D}}, \vec{D} \big) \le \frac{\epsilon^2}{8} 
    ~.
\]

Putting together with Lemma~\ref{lem:hellinger-and-total-variation} proves Eqn.~\eqref{eqn:strong-monotone-tv-bound}, which finishes the proof of Theorem~\ref{thm:strong-monotone-upper-bound}.

\begin{proof}[Proof of Lemma~\ref{lem:strong-monotone-Dhat-and-D} and Lemma~\ref{lem:strong-monotone-Dcheck-and-D}]
By symmetry, it suffices to prove one of them.
Below we prove Lemma~\ref{lem:strong-monotone-Dhat-and-D}.

For simplicity of notations in this proof, let $\Gamma = \frac{\ln(2Nn\delta^{-1})}{N}$ be the coefficient that appears in the definition of $\vec{\hat{D}}$, i.e., Eqn.~\eqref{eqn:strong-monotone-Dhat}.
Further define:
\[
    g(y) = y - \sqrt{2\Gamma \cdot y(1-y)} - \Gamma 
    ~.
\]

Then, Eqn.~\eqref{eqn:strong-monotone-Dhat} can be written as:
\begin{equation}
    \label{eqn:strong-monotone-Dhat-rewrite}
    F_{\hat{D}_i}(\type) = \begin{cases}
        1 & x = 1 ~; \\
        \max \big\{ 0, g\big( F_{D_i}(\type) \big) \big\} & 0 \le x < 1 ~.
    \end{cases}
\end{equation}

Further, the inequality in the lemma can be written as:
\[
    \Hellinger^2 \big( \vec{\hat{D}}, \vec{D} \big) \le O \left( n \Gamma \log \frac{1}{\Gamma} \right)
    ~,
\]
or equivalently:
\[
    1 - \Hellinger^2 \big( \vec{\hat{D}}, \vec{D} \big) \ge 1 - O \left( n \Gamma \log \frac{1}{\Gamma} \right)
    ~.
\]

By Lemma~\ref{lem:hellinger-decompose}, we have:
\[
    1 - \Hellinger^2 \big( \vec{\hat{D}}, \vec{D} \big) = \prod_{i=1}^n \big( 1 - \Hellinger^2 \big( \hat{D}_i, D_i \big) \big)
    ~.
\]

Hence, it suffices to show that for any $i \in [n]$:
\[
    1 - \Hellinger^2 \big( \hat{D}_i, D_i \big) \ge 1 - O \left( \Gamma \log \frac{1}{\Gamma} \right)
    ~,
\]
or equivalently:
\[
    \Hellinger^2 \big( \hat{D}_i, D_i \big) \le O \left( \Gamma \log \frac{1}{\Gamma} \right)
    ~.
\]

By definition, the squared Hellinger distance is:
\begin{equation}
    \label{eqn:strong-monotone-squared-hellinger}
    \Hellinger^2 \big( \hat{D}_i, D_i \big) = \frac{1}{2} \int_{x \in [0, 1]} \left( \sqrt{dF_{D_i}(x)} - \sqrt{dF_{\hat{D}_i}(x)} \right)^2
    ~.
\end{equation}

We shall partition $[0, 1]$ into three subsets based on how the CDF of $\hat{D}_i$ is defined in Eqn.~\eqref{eqn:strong-monotone-Dhat-rewrite}: (a) the values whose $F_{\hat{D}_i}(\type) = 0$, (b) $x = 1$ whose $F_{\hat{D}_i}(\type) = 1$, and (c) the rest of the values whose $0 < F_{\hat{D}_i}(\type) < 1$.
Then, we account for their contributions to Eqn.~\eqref{eqn:strong-monotone-squared-hellinger} separately. 

\paragraph{Part (a).}
Consider the values whose CDF is $0$ w.r.t.\ $\hat{D}_i$.
To formally define this subset of values, recall that $g(y) = y - \sqrt{2\Gamma \cdot y(1-y)} - \Gamma$.
Let $F_\ell \in [0, 1]$ be the unique solution for: 
\[
    g(F_\ell) = 0
    ~.
\]

The value of $g(F_\ell)$ is strictly less than $0$ when $F_\ell = \Gamma$, and is strictly greater than $0$ when $F_\ell = 4 \Gamma$.
Hence, we have:
\begin{equation}
    \label{eqn:strong-monotone-cdf-lower-bound}
    \Gamma < F_\ell < 4 \Gamma
    ~.
\end{equation}

Let $\ell$ be the minimum value whose CDF is at least $F_\ell$, i.e.:
\[
    \ell = \inf \big\{ x : F_{D_i}(\type) \ge F_\ell \big\}
    ~.
\]

Then, for values in $[0, \ell)$, we have $F_{\hat{D}_i}(\type) = 0$ and therefore:
\begin{align*}
    \frac{1}{2} \int_{\type \in [0, \ell)} \left( \sqrt{dF_{D_i}(\type)} - \sqrt{dF_{\hat{D}_i}(\type)} \right)^2 
    &
    = \lim_{\type \to \ell-} \frac{1}{2} F_{D_i}(\type) \\
    &
    \le \frac{1}{2} F_\ell 
    && \text{(definition of $\ell$)} \\[2ex]
    &
    < 2 \Gamma
    ~. 
    && \text{(Eqn.~\eqref{eqn:strong-monotone-cdf-lower-bound})}
\end{align*}

\paragraph{Part (b).}
For simplicity of notations, let $f(1) = f_{D_i}(1)$ and $\hat{f}(1) = f_{\hat{D}_i}(1)$ be the probability that $x = 1$ in $D_i$ and $\hat{D}_i$ respectively.
We have: 
\begin{align*}
    \hat{f}(1)
    & 
    = 1 - \lim_{\type \to 1-} F_{\hat{D}_i}(\type) \\
    &
    = 1 - \lim_{\type \to 1-} \bigg( F_{D_i} (\type) - \sqrt{2 \Gamma \cdot F_{D_i}(\type) \big( 1-F_{D_i}(\type) \big)} - \Gamma \bigg) 
    && \text{(Eqn.~\eqref{eqn:strong-monotone-Dhat-rewrite})} \\
    & 
    = f(1) + \sqrt{2 \Gamma \cdot f(1) \big(1-f(1)\big)} + \Gamma
    ~.
\end{align*}

As corollaries, we have:
\[
    \hat{f}(1) \ge \max \big\{ f(1), \Gamma \big\} 
    ~,
\]
and:
\begin{align*}
    \big( \hat{f}(1) - f(1) \big)^2
    &
    = \Gamma \cdot \big( \sqrt{2 f(1) \big(1-f(1)\big)} + \sqrt{\Gamma} \big)^2 \\
    &
    \le \Gamma \cdot \big( \sqrt{2 f(1)} + \sqrt{\Gamma} \big)^2 \\[1ex]
    &
    \le \Gamma \cdot \max \big\{ f(1), \Gamma \big\}
    ~.
\end{align*}

Using the above two inequalities, the contribution from $x = 1$ is at most:
\begin{align*}
    \frac{1}{2} \big( \sqrt{f(1)} - \sqrt{\hat{f}(1)} \big)^2
    & 
    = \frac{1}{2} \frac{\big( \hat{f}(1) - f(1) \big)^2}{\big( \sqrt{f(1)} + \sqrt{\hat{f}(1)} \big)^2} \\
    &
    \le \frac{\big( \hat{f}(1) - f(1) \big)^2}{2 \hat{f}(1)} \\[1ex]
    &
    \le \frac{\Gamma}{2}
    ~.
\end{align*}

\paragraph{Part (c).}
It remains to bound the contribution from values $\type \in [\ell, 1)$.
By Eqn.~\eqref{eqn:strong-monotone-Dhat-rewrite} and the definition of $\ell$, the CDF w.r.t.\ $\hat{D}_i$ of any value in this range is defined by a continuous mapping:
\[
    F_{\hat{D}_i}(\type) = g \big( F_{D_i}(\type) \big)
    ~.
\]

Therefore, the CDFs w.r.t.\ $D_i$ and $\hat{D}_i$ have the same set of discontinuities in this range, i.e., the same set of point masses.
We will first bound the contribution of values in $[\ell, 1)$ under the assumption that both $D_i$ and $\hat{D}_i$ are continuous in this range.
Then, we will demonstrate how to generalize the result to arbitrary distributions by handling the common point masses appropriately.

Under the assumption of continuity, the contribution to the Hellinger distance by this part is:
\begin{align*}
    \frac{1}{2} \int_{\type \in [\ell, 1)} \left( \sqrt{dF_{D_i}(\type)} - \sqrt{dF_{\hat{D}_i}(\type)} \right)^2 
    & 
    = \frac{1}{2} \int_{\type \in [\ell, 1)} \bigg( \sqrt{\frac{dF_{\hat{D}_i(\type)}}{dF_{D_i}(\type)}} - 1 \bigg)^2 d F_{D_i}(\type) \\
    &
    = \frac{1}{2} \int_{\type \in [\ell, 1)} \bigg( \sqrt{g' \big( F_{D_i}(\type) \big)} - 1 \bigg)^2 d F_{D_i}(\type)
    ~.
\end{align*}

By the definition of $g$, we have:
\[
    g'(y) = 1 + \frac{(2 y - 1) \sqrt{\Gamma}}{\sqrt{2 y \big( 1 - y \big)}}
    ~.
\]

Therefore, it suffices to upper bound the following integral:
\[
    \int_{F_\ell}^{1} \bigg( \sqrt{1 + \tfrac{(2 y - 1) \sqrt{\Gamma}}{\sqrt{2 y (1 - y)}}} - 1 \bigg)^2 dy 
\]

We will bound it in $[F_\ell, 1 - \Gamma)$ and $[1 - \Gamma, 1]$ separately.
The former is at most:
\begin{align*}
    \int_{F_\ell}^{1 - \Gamma} \bigg( \sqrt{1 + \tfrac{(2 y - 1) \sqrt{\Gamma}}{\sqrt{2 y \big( 1 - y \big)}}} - 1 \bigg)^2 dy 
    &
    \le \int_{F_\ell}^{1 - \Gamma} \frac{(2 y - 1)^2 \Gamma}{2 y \big( 1 - y \big)} dy 
    && \text{($\big|\sqrt{1+x} - 1\big| \le \big|x\big|$)} \\
    &
    \le \int_{F_\ell}^{1 - \Gamma} \frac{\Gamma}{2 y \big( 1 - y \big)} dy 
    && \textrm{($0 \le y \le 1$)} \\[1ex]
    &
    = \frac{\Gamma}{2} \bigg( \ln \frac{1 - \Gamma}{F_\ell} + \ln \frac{1 - F_\ell}{\Gamma} \bigg) \\[1ex]
    &
    \le \frac{\Gamma}{2} \bigg( \ln \frac{1}{F_\ell} + \ln \frac{1}{\Gamma} \bigg) \\[1ex]
    &
    < \Gamma \ln \frac{1}{\Gamma}
    ~.
    && 
    \text{(Eqn.~\eqref{eqn:strong-monotone-cdf-lower-bound})}
    ~.
\end{align*}

For the latter, it is upper bounded by:
\begin{align*}
    \int_{1 - \Gamma}^1 \bigg( \sqrt{1 + \tfrac{(2 y - 1) \sqrt{\Gamma}}{\sqrt{2 y \big( 1 - y \big)}}} - 1 \bigg)^2 dy 
    &
    \le \int_{1 - \Gamma}^1 \tfrac{(2 y - 1) \sqrt{\Gamma}}{\sqrt{2 y \big( 1 - y \big)}} dy 
    && \text{($\sqrt{1 + x} - 1 \le \sqrt{x}$ for $x > 0$)} \\
    &
    = \sqrt{2(1-\Gamma)} \Gamma \\[2ex]
    &
    \le \sqrt{2} \Gamma
    ~.
\end{align*}

Combining the upper bounds of the integrals over the two intervals, the contribution from part (c) under the assumption of continuity is at most $O\big( \Gamma \log \frac{1}{\Gamma}\big)$.

Finally, consider any point mass $\type^*$ in the two distributions $D_i$ and $\hat{D}_i$.
Let $\bar{y} = F_{D_i}(\type^*)$ and $\underline{y} = \lim_{\type \to \type^*-} F_{D_i}(\type)$.
Then, the probability mass of $\type^*$ w.r.t.\ $D_i$ is $\bar{y} - \underline{y}$, and that w.r.t.\ $\hat{D}_i$ is $g(\bar{y}) - g(\underline{y})$.
Hence, the contribution of $\type^*$ to the Hellinger distance is:
\begin{align*}
    \frac{1}{2} \bigg( \sqrt{\bar{y} - \underline{y}} - \sqrt{g(\bar{y}) - g(\underline{y})} \bigg)^2
    &
    = 
    \frac{1}{2} \bigg( \sqrt{\frac{g(\bar{y}) - g(\underline{y})}{\bar{y} - \underline{y}}} - 1 \bigg)^2 \big( \bar{y} - \underline{y} \big) \\
    &
    = \frac{1}{2} \bigg( \sqrt{\frac{1}{\bar{y} - \underline{y}} \int_{\underline{y}}^{\bar{y}} g'(y) dy } - 1 \bigg)^2 \big( \bar{y} - \underline{y} \big) \\
    &
    \le \frac{1}{2} \int_{\underline{y}}^{\bar{y}} \bigg( \sqrt{g'(y)} - 1 \bigg)^2 dy
    ~.
\end{align*}

The last inequality follows by the convexity of $(\sqrt{x} - 1)^2$ and Jensen's inequality.
The RHS is precisely the contribution by values with CDF in $(\underline{y}, \bar{y}]$ in the continuous case.
Therefore, by applying this argument to all point masses, we reduce the problem to the continuous case.
\end{proof}

\subsection{Applications}

\paragraph{Single-parameter Revenue Maximization.}
As we have discussed at the beginning of the section, strong monotonicity corresponds to strong revenue monotonicity in the context of single-parameter revenue maximization, which is shown by \citet{DevanurHP/STOC/2016}.
In particular, for single-item auction, it follows from Theorem~\ref{thm:strong-monotone-upper-bound} that $\tilde{O}(n \epsilon^{-2})$ samples are sufficient for getting an $\epsilon$-additive approximation when the bidders' valuations are bounded in $[0, 1]$, matching the optimal bound by \citet{guo2019settling}.
The main difference compared with \citet{guo2019settling} lies in that we achieve the optimal upper bound using the empirical maximizer, which corresponds to Myerson's optimal auction w.r.t.\ the empirical distributions, while \citet{guo2019settling} needs to apply appropriate regularization to the empirical distribution and uses the corresponding regularized empirical Myerson's auction.

\begin{theorem}
    \label{thm:single-item-auction}
    In a single-item auction with $n$ bidders whose values are bounded in $[0, 1]$, suppose the number of samples is at least:
    \[
        C \cdot \frac{n}{\epsilon^2} \log \bigg( \frac{n}{\epsilon} \bigg) \log \bigg( \frac{n}{\epsilon \delta} \bigg)
    \]
    for some sufficiently large constant $C > 0$.
    Then, the empirical Myerson's auction is an $\epsilon$-additive approximation with probability at least $1 - \delta$.
\end{theorem}

\paragraph{Prophet Inequality.}
In the context of prophet inequality, each hypothesis corresponds to a sequence of thresholds, one for each round, such that the algorithm accepts the first reward that is greater than or equal to the corresponding threshold.
We show that this problem satisfies strong monotonicity;
the proof is deferred to Appendix~\ref{app:strong-monotone-prophet}.

\begin{lemma}
    \label{lem:strong-monotone-prophet}
    The problem of prophet inequality is strongly monotone.
\end{lemma}

As a corollary of Lemma~\ref{lem:strong-monotone-prophet}, Theorem~\ref{thm:strong-monotone-upper-bound}, and the fact that the optimal thresholds achieve at least one half of the expected max reward (e.g., \citet{Samuel-Cahn/AnnProbab/1984}),%
\footnote{In fact, is it known that an appropriate fixed threshold can achieve at least one half of the expected max.}
we get an $\tilde{O} \big( n \epsilon^{-2} \big)$ sample complexity upper bound.

\begin{theorem}
    \label{thm:prophet}
    For any instance of prophet inequality in which the rewards are bounded in $[0, 1]$, suppose the number of samples is at least:
    \[
        C \cdot \frac{n}{\epsilon^2} \log \bigg( \frac{n}{\epsilon} \bigg) \log \bigg( \frac{n}{\epsilon \delta} \bigg)
    \]
    for some sufficiently large constant $C > 0$.
    Then, the expected reward by the empirically optimal thresholds is an $\epsilon$-additive approximation compared to the optimal thresholds and thus, is at least half of the expected max reward minus $\epsilon$.
\end{theorem}

\paragraph{Prophet Inequality for i.i.d.\ Rewards.}
If the rewards are i.i.d., \citet{CorreaFHOV/EC/2017} prove an improved prophet inequality that achieves at least a $0.745$ factor of the expected max reward.
The strong monotonicity of prophet inequality for i.i.d.\ rewards follows as a special case of Lemma~\ref{lem:strong-monotone-prophet}.
Therefore, we get the same $\tilde{O} \big( n \epsilon^{-2} \big)$ sample complexity upper bound, matching the lower bound by \citet{correa2018prophet}.

\begin{theorem}
    \label{thm:iid-prophet}
    For any instance of prophet inequality with i.i.d.\ rewards bounded in $[0, 1]$, suppose the number of sample rewards (rather than reward vectors) is at least:
    \[
        C \cdot \frac{n}{\epsilon^2} \log \bigg( \frac{n}{\epsilon} \bigg) \log \bigg( \frac{n}{\epsilon \delta} \bigg)
    \]
    for some sufficiently large constant $C > 0$.
    Then, the expected reward by the empirically optimal thresholds is an $\epsilon$-additive approximation compared to the optimal thresholds and thus, is at least a $0.745$ factor of the expected max reward minus $\epsilon$.
\end{theorem}

As mentioned in Section~\ref{sec:introduction}, the setting of \citet{correa2018prophet} is different from ours in that they consider unbounded distributions and multiplicative approximation.
Indeed, we focus on bounded-support distributions and additive approximation in the main text of the paper in order to develop a generalization theory that requires minimum knowledge of the structure of the problems.
Nonetheless, Appendix~\ref{app:prophet-unbounded} demonstrates how to combine the techniques in this paper and the special structures of the prophet inequality with i.i.d.\ rewards to get the same $\tilde{O}(n \epsilon^{-2})$ sample complexity upper bound in the setting of \citet{correa2018prophet}, addressing an open problem therein.%
\footnote{It is explicitly stated as an open question in the talk at EC 2019.}

\paragraph{Pandora's Problem.}
An algorithm for the Pandora's problem is a mapping from the history of observed rewards to either one of the unopened boxes, or the decision to stop and take the best observed reward.
Since the former has exponentially many possibilities even after discretization, the na\"{i}ve upper bound on the size of the hypothesis class is doubly exponential.
Nonetheless, we show that the problem is strongly monotone, highlighting that strong monotonicity is a structural property without any obvious connection to the complexity/simplicity of the hypothesis class.
The proof is deferred to Appendix~\ref{app:strong-monotone-pandora}.

\begin{lemma}
    \label{lem:strong-monotone-pandora}
    Pandora's problem is strongly monotone.
\end{lemma}

Recall that in Section~\ref{sec:preliminary} we use the simple treatment of defining the value of a hypothesis to be the value of the corresponding algorithm plus $n$ and then scaled by $\frac{1}{n+1}$ to normalize its range to be $[0, 1]$.
Therefore, to get an $\epsilon$-additive approximation in Pandora's problem, we need a $\frac{\epsilon}{n+1}$-additive approximation w.r.t.\ $\hypoclass$.
As a corollary of Lemma~\ref{lem:strong-monotone-pandora} and Theorem~\ref{thm:strong-monotone-upper-bound}, we get an $\tilde{O}(n^3 \epsilon^{-2})$ sample complexity bound.
See Appendix~\ref{app:pandora-improved} for an analysis tailored for Pandora's problem to get the following optimal bound.

\begin{theorem}
    \label{thm:pandora}
    For any instance of Pandora's problem in which the rewards are bounded in $[0, 1]$, suppose the number of samples is at least:
    \[
        C \cdot \frac{n}{\epsilon^2} \log^2 \bigg( \frac{1}{\epsilon} \bigg) \log \bigg( \frac{n}{\epsilon} \bigg) \log \bigg( \frac{n}{\epsilon \delta} \bigg)
    \]
    for some sufficiently large constant $C > 0$.
    Then, we can learn an $\epsilon$-additive approximate algorithm from the samples.
    Further, to learn such an algorithm, the number of samples must be at least:
    \[
        c \cdot \frac{n}{\epsilon^2}
    \]
    for some sufficiently small constant $c > 0$.
\end{theorem}

\section{Classification Problems: a Preliminary Discussion}
\label{sec:classification}

In classification problems, there is a special data dimension which corresponds to the labels;
the rest of the data dimensions correspond to the features.
In particular, it is crucial that the labels are correlated with the features.  
Therefore, the assumption of independent data dimensions fail to hold.
Nevertheless, below we present a straightforward extension of Theorem~\ref{thm:error_finite} under the assumption that the distribution of features \emph{conditioned on any given label} is a product distribution.
Although this preliminary result still relies on too strong an assumption to be useful in natural classification problems, we hope that it will serve as a stepping stone for follow-up works. See Section~\ref{sec:future} for some related research directions.

The rest of the section follows the notations in classification problems and denotes each data point as a feature-label pair $(\vec{x}, y)$, where $\vec{x}$ is the feature vector and $y$ is the label.
We assume that there are $\ell$ labels $[\ell] = \{1, 2, \dots, \ell\}$.
Let $\vec{T} = \prod_{i=1}^n T_i$ denote an $n$-dimensional feature domain.
Hence, the data domain under the model in Section~\ref{sec:preliminary} is $\vec{T} \times [\ell]$.
Let $D_Y$ denote the distribution of labels.
Further, for any label $y \in [\ell]$, let $\vec{D}_{\vec{X}|y}$ denote a \emph{product} distribution of features conditioned on having label $y$.
For simplicity of notation, let $\vec{D}_{\vec{X}}$ denote the collection of conditional product feature distributions, and write $\vec{D} = \vec{D}_{\vec{X}} \circ D_Y$ be the joint distribution of feature-label pairs.
%
%
%
By definition, the probability mass function of the joint distribution is:
\begin{equation}
    \label{eqn:classification-joint-mass}
    f_{\vec{D}}(\mathbf{x}, y) = f_{D_Y}(y) \cdot f_{\vec{D}_{\vec{X}|y}} (\vec{x})
    ~.
\end{equation}

We say that such a distribution has product conditional feature distributions.

\paragraph{Generalized Product Empirical Distribution.}
We now generalize the definition of product empirical distribution to classification problems that have product conditional feature distributions.
Let the empirical distribution of labels $E_Y$ be the uniform distribution over sample labels.
Further, for any label $y \in [\ell]$, let $\vec{E}_{\vec{X}|y}$ be the product empirical distribution w.r.t.\ the samples with label $y$.
Concretely, for any $i \in [n]$, let the $i$-th coordinate of $\vec{E}_{\vec{X}|y}$ be a uniform distribution over the $i$-th coordinate of the samples with label $y$.
As before, let $\vec{E}_{\vec{X}}$ denote the collection of product empirical feature distributions conditioned on the labels.
Finally, let $\vec{E} = \vec{E}_{\vec{X}} \circ E_Y$.

By definition, the probability mass function of the joint distribution is:
\begin{equation}
    \label{eqn:classification-joint-mass-empirical}
    f_{\vec{E}}(\mathbf{x}, y) = f_{E_Y}(y) \cdot f_{\vec{E}_{\vec{X}|y}} (\vec{x})
    ~.
\end{equation}

Finally, we define the \emph{product empirical risk minimizer} (PERM) to be the best hypothesis w.r.t.\ $\vec{E}$.
Here, note that we seek to minimize the objective.

\begin{theorem}
    \label{thm:classification}
    Let $\vec{D}= \vec{D}_{\vec{X}} \circ D_Y$ be any distribution with product conditional feature distributions, over $\vec{\typespace} \times [\ell]$ such that $|T_i|\le k$ for any $1\le i\le n$. 
    For a sufficiently large constant $C>0$, suppose the number of samples is at least:
    \[
        C \cdot \frac{nk\ell}{\epsilon^2}\log \left(\frac{\ell}{\delta}\right)
    \]
    Then, with probability at least $1-\delta$, for any $\hypo : \vec{T} \times [\ell] \mapsto [0,1]$, we have:
    \begin{equation*}
        \big| \hypo (\vec{D}) - \hypo(\vec{E}) \big| \le \epsilon
        ~.
    \end{equation*}
    In particular, the PERM is an $\epsilon$-additive approximation.
\end{theorem}

\begin{proof}
    Similar to the proof of Theorem~\ref{thm:error_finite}, we rely on Lemma~\ref{lem:total-variation}.
    It suffices to show that:
    \[
        \TV (\mathbf{D}, \mathbf{E}) \le \epsilon
        ~.
    \]

    To do so, we first decompose it into two parts, the error due to the estimation of the label distribution, and that due to the conditional feature distributions. 
    By Eqn.~\eqref{eqn:classification-joint-mass} and Eqn.~\eqref{eqn:classification-joint-mass-empirical}:
    \begin{align*}
        \TV (\mathbf{D}, \mathbf{E})  
        & 
        = \frac{1}{2} \sum_{y \in [\ell]} \sum_{\vec{x} \in \vec{\typespace}} \big| f_{D_Y}(y) \cdot f_{\vec{D}_{\vec{X}|y}} (\vec{x}) - f_{E_y}(y) \cdot f_{\vec{E}_{\vec{X}|y}}(\vec{x}) \big| \\
        &
        \le \frac{1}{2} \sum_{y \in [\ell]} \sum_{\vec{x}\in \vec{\typespace}} \left( \big| f_{D_y}(y) - f_{E_y}(y) \big| \cdot f_{\vec{D}_{\vec{X}|y}}(\vec{x}) + f_{E_y}(y) \cdot \big| f_{\vec{D}_{\vec{X}|y}}(\vec{x}) - f_{\vec{E}_{\vec{X}|y}}(\vec{x}) \big| \right) \\
        & 
        = \TV(D_y, E_y) + \sum_{y \in [\ell]} f_{E_y}(y) \cdot \TV(\vec{D}_{\vec{X}|y}, \vec{E}_{\vec{X}|y})
        ~.
    \end{align*}

    By Lemma~\ref{lem:hellinger_distance} and the stated number of samples, the squared Hellinger distance between the label distributions $D_y$ and $E_y$ is less than $\frac{\epsilon^2}{8}$.
    Further by the relation between the total variation and Hellinger distances, i.e., Lemma~\ref{lem:hellinger-and-total-variation}, we get that the first term on the RHS above is at most $\frac{\epsilon}{2}$.

    It remains to bound the second term, i.e., the error due to the estimation of the conditional feature distributions.
    Fix any label $y \in [\ell]$.
    By definition, the number of samples with label $y$ is $f_{E_Y}(y) N$.
    Therefore, by Lemma~\ref{lem:hellinger_distance} and the stated number of samples, the squared Hellinger distance between the feature distributions conditioned on $y$ is at most:
    \[
        \Hellinger^2 ( \vec{D}_{\vec{X}|y}, \vec{E}_{\vec{X}|y} ) \le \frac{\epsilon^2}{8 \ell f_{E_Y}(y)}        ~.
    \]

    Further by Lemma~\ref{lem:hellinger-and-total-variation}, their total variation distance is at most:
    \[
        \TV ( \vec{D}_{\vec{X}|y}, \vec{E}_{\vec{X}|y} ) \le \frac{\epsilon}{2} \cdot \frac{1}{\sqrt{\ell f_{E_Y}(y)}}
        ~.
    \]

    Hence, the second term is at most:
    \begin{align*}
        \sum_{y \in [\ell]} f_{E_Y}(y) \cdot \TV(\vec{D}_{\vec{X}|y}, \vec{E}_{\vec{X}|y})
        &
        \le \sum_{y \in [\ell]} \frac{\epsilon}{2} \cdot \sqrt{\frac{f_{E_Y}(y)}{\ell}} \\
        &
        \le \frac{\epsilon}{2}
        ~.
    \end{align*}
        
    The second inequality follows by $\sum_{y \in [\ell]} f_{E_Y}(y) = 1$ and the Cauchy-Schwartz inequality.
    %
\end{proof}

\section{Future Directions and Other Related Works}
\label{sec:future}

\paragraph{Sample Complexity of Simple Auctions/Hypotheses.}
A branch of the literature of sample complexity of auctions considers simpler auction formats such as the second-price auction with reserve prices. 
Readers are referred to \citet{BalcanSV/EC/2018}, \citet{CaiD/FOCS/2017}, and \citet{MorgensternR/COLT/2016} for some examples.
We restate that the theories developed in this paper are complementary to the existing ones;
they are more suitable for problems with complex hypothesis classes (on product distributions).
Hence, this paper does not try to apply the theories to these simpler families of auctions. 
That said, there are relatively few natural hypothesis classes whose ``degrees-of-freedom'' are smaller than the data dimensions.
Hence, for strongly monotone problems, the sample complexity bound in Theorem~\ref{thm:strong-monotone-upper-bound} is competitive.
Finally, we leave as a future question whether there are natural learning problems whose tight sample complexity bounds need both complexity measures of the hypotheses and the independence of data dimensions.

\paragraph{Beyond Product Distributions.}
Although arbitrarily correlated distributions seem intractable, it may be possible to generalize the theories in this paper to structured corrected distributions, which we leave as another future direction.
Concretely, if we can learn from samples an appropriate representation of the data under which different dimensions are independent, we shall be able to combine it with the theories in this paper to get generalization bounds.
To this end, the vast literature on principle component analysis (PCA) is related.
See, e.g., \citet{Pearson/1901} and \citet{Jolliffe/book/2011}.
Independently, \citet{ec20/BrustleCD2019} made progress on this direction showing how to learn multi-item auctions when the value distribution is correlated yet admit special structures.

\paragraph{Classification Problems.}
%
Section~\ref{sec:classification} presents a preliminary result under a strong assumption that the feature distributions are independent \emph{conditioned on any given label}.
To further extend the theories in this paper to obtain useful generalization bounds for natural classification problems, we need to relax the assumption of having product conditional feature distributions, which is related to the last research direction.
Moreover, although the algorithmic question of finding the optimal hypothesis w.r.t.\ a product distribution is well-studied for optimization problems in the Bayesian model, little is known about its counterpart for classification problems.
In particular, it is unclear whether finding the best hypothesis w.r.t.\ the product empirical distribution is harder or easier than doing so w.r.t.\ the original notion of empirical distribution.
On the one hand, the product empirical distribution is more structured;
on the other hand, its support size is exponential in general, while the support size of the original empirical distribution is upper bounded by the number of samples.

\paragraph{Multi-parameter Auctions and Other Structural Properties.}
Multi-parameter revenue maximization is the only example in this paper that does not benefit from the improved sample complexity bound in Theorem~\ref{thm:strong-monotone-upper-bound} because it is not strong monotonicity.
In fact, \citet{HartR/TE/2015} showed that it is not even weakly monotone.
Nonetheless, the hypotheses corresponding to multi-parameter auctions are very different from those used in the proof of the lower bound (Theorem~\ref{thm:finite_domain_lower_bound}).
We consider it an interesting open question if there is another structural property (unrelated to the complexity measures of the hypotheses) which applies to multi-parameter revenue maximization and, ideally, to a large family of problems, which lead to improved sample complexity bounds.

\bibliographystyle{plainnat}

\bibliography{ref}

\appendix

\section{Lower Bound for Finite-domain Problems: Proof of Theorem~\ref{thm:finite_domain_lower_bound}}
\label{app:finite-lb}

For notation simplicity, we consider $T=\{0, \pm 1, \ldots, \pm k\}$ with support size $2k+1$. We first define the hypothesis class $\hypoclass$. Each hypothesis is specified by a binary $nk$-dimensional vector $\vec{v} \in \{\pm 1\}^{n\times k}$. Specifically, $\hypoclass = \{ \hypo^{\vec{v}}\}$ where $\hypo^{\vec{v}} : \typespace^n \to [0,1]$ is defined as
\[
\hypo^{\vec{v}}(\vec{\type}) := \mathbbm{1} \left[ \exists i \in [n],j \in [k], \vec{\type} = \left( 0,\ldots,0,\underset{i\text{'th}}{v_{i,j}\cdot j}, 0, \ldots,0 \right) \right].
\]

Next, we consider a family of distributions $\mathcal{D} = \{\vec{D}^{\vec{v}} \}$ that are also indexed by $\vec{v}$. For each dimension $i$ of $\vec{D}^{\vec{v}}$, the probability density function is defined as the following:
\[
f_{D^{\vec{v}}_i}(\type_i)= 
\begin{cases} 1-\frac{1}{n}, & \text{if } \type_i=0 \\
\frac{1}{2nk}(1-\epsilon), & \text{if } \type_i=-v_{i,j} \cdot j \text{ for some } j\in[k] \\
\frac{1}{2nk}(1+\epsilon), & \text{if } \type_i=v_{i,j} \cdot j  \text{ for some } j\in[k]
\end{cases}
\]
Our plan is to show that any algorithm that gets a $\epsilon$-approximation on all distributions in $\mathcal{D}$ must take $\Omega( \frac{nk}{\epsilon^2})$ number of samples.

When the underlying distribution is $\vec{D}^{\vec{v}}$, the corresponding optimal hypothesis is $h^{\vec{v}}$. Intuitively, in order to achieve a good approximation to $h^{\vec{v}}$, an algorithm has to specify a vector $\vec{v}'$ close enough to $\vec{v}$ based on the samples. We formalize the intuition by calculating the loss of choosing $h^{\vec{v}'}$.

\begin{lemma}\label{lem:finite_lb_error}
For all $\vec{v}, \vec{v'}$, we have
\[
    h^{\vec{v}}(\vec{D}^{\vec{v}}) - h^{\vec{v}'}(\vec{D}^{\vec{v}}) = \Omega(\frac{\epsilon}{nk}\cdot d(\vec{v},\vec{v'})),
\]
where $d(\vec{v},\vec{v}')$ is the hamming distance between $\vec{v}$ and $\vec{v}'$.
\end{lemma}
\begin{proof}
For all $\vec{v}, \vec{v}'$, we have that
    \begin{align*}
        h^{\vec{v}}(\vec{D}^{\vec{v}'}) &= \sum_{i=1}^n \sum_{j=1}^{k}\Pr_{\type_i\sim D^{\mathbf{v}}_i}[\type_i=v'_{i,j} \cdot j] \cdot \prod_{\ell \neq i}\Pr_{\type_\ell\sim D_\ell^{\mathbf{v}}}[\type_\ell=0]\\
        &=(1-\frac{1}{n})^{n-1}\sum_{i=1}^n \sum_{j=1}^{k}\Pr_{\type_i\sim D^{\vec{v}}_i}[\type_i=v'_{i,j} \cdot j]\\
        &=(1-\frac{1}{n})^{n-1}\sum_{i=1}^n \sum_{j=1}^{k} \left[\mathbbm{1}[v_{i,j}=v'_{i,j}]\cdot \frac{1+\epsilon}{2nk} + \mathbbm{1}[v_{i,j} \neq v'_{i,j}]\cdot \frac{1-\epsilon}{2nk} \right] ~.
    \end{align*}
Next, we bound the lose of choosing $h^{\vec{v}'}$ by the hamming distance between $\mathbf{v}$ and  $\mathbf{v'}$.
\[
h^{\vec{v}}(\vec{D}^{\vec{v}}) - h^{\vec{v}'}(\vec{D}^{\vec{v}}) = (1-\frac{1}{n})^{n-1}\sum_{i=1}^n \sum_{j=1}^{k} \mathbbm{1}[v_{i,j}\neq v'_{i,j}] \cdot \frac{\epsilon}{nk}
= \Omega(\frac{\epsilon}{nk}\cdot d(\vec{v},\vec{v'}))
\]
\end{proof}

Let $\vec{s}$ be the samples and $A$ be any (randomized) algorithm that takes samples $\vec{s}$ as inputs and outputs a vector $A(\vec{s}) \in \{\pm 1\}^{n\times k}$. The next lemma states that if two distributions differ in only one dimension, then the total probability of $A$ guessing wrongly for the two distributions is at least $\Omega(1)$ if the number of samples is $O(\frac{nk}{\epsilon^2})$.

\begin{lemma}\label{lem:finite_lb_error_prob}
    For any $\vec{D^{\vec{\bar{v}}}}$ and  $\vec{D^{\vec{\underline{v}}}}$ where $\vec{D^{\vec{\bar{v}}}}$ and  $\vec{D^{\vec{\underline{v}}}}$ only differ in one dimension $(i,j)$, i.e., $\bar{v}_{i,j}=1$, $\underline{v}_{i,j}=-1$, and for any algorithm $A$, when $N = O(\frac{nk}{\epsilon^2})$,
    \begin{equation*}
        \Pr_{\vec{s}\sim (\vec{D}^\vec{\bar{v}})^N} [A(\vec{s})_{i,j} \ne \bar{v}_{i,j}] + \Pr_{\vec{s}\sim (\vec{D}^\vec{\underline{v}})^N} [A(\vec{s})_{i,j} \ne \underline{v}_{i,j}] \ge \Omega(1).
    \end{equation*}
\end{lemma}

\begin{proof}
Since the $n$ dimensions of the distribution are independent, to guess the $(i,j)$-th dimension of the underlying $\vec{v}$, the only useful samples are those $\vec{s}$ with $s_i=\pm j$, which happens with probability $\frac{1}{nk}$. Thus, when the number of samples is $O(\frac{nk}{\epsilon^2})$, with high probability, we have at most $O(\frac{1}{\epsilon^2})$ number of useful samples.
Note that these samples are either from 
\[
\bar{f}(\pm j) = \frac{1 \pm \epsilon}{2} \text{ or }  \underline{f}(\pm j) = \frac{1 \mp \epsilon}{2}.
\]
whose total variation distance is $\Theta(\epsilon)$, then if we only have $O(\frac{1}{\epsilon^2})$ samples, with constant probability we cannot distinguish distinguish whether the samples are from $\vec{D}^{\vec{\bar{v}}}$ or $\vec{D}^{\vec{\underline{v}}}$. In other words, for any algorithm $A$, $A(\vec{s})_{i,j}$ must be inconsistent with the underlying distribution with constant probability, which concludes the proof.
\end{proof}



To finish the proof, we consider the performance of $A$ on a uniform distribution over all distributions in $\mathcal{D}$. Formally, let $U$ be a uniform distribution on $\{\pm 1\}^{n\times k}$, we have
\begin{align*}
    \E_{\vec{v}\sim U}\E_{\vec{s}\sim (\mathbf{D^v})^N }[d(A(\vec{s}),\vec{v})] & = \sum_{i=1}^n\sum_{j=1}^{k}\E_{\vec{v}\sim U}\E_{\vec{s}\sim (\mathbf{D^v})^N } [\mathbbm{1}(A(\vec{s})_{i,j} \ne v_{i,j})] \\
    & = \sum_{i=1}^n\sum_{j=1}^{k}\E_{\vec{v}\sim U}\Pr_{\vec{s}\sim (\mathbf{D^v})^N } [A(\vec{s})_{i,j}) \ne v_{i,j}] \\
    & = \frac{1}{2} \sum_{i=1}^n\sum_{j=1}^{k} \Big( \E_{\vec{\bar{v}}\sim U}\Pr_{\vec{s}\sim (\mathbf{D}^\vec{\bar{v}})^N} [A(\vec{s})_{i,j}) \ne \bar{v}_{i,j}] \\
    & \phantom{=\frac{1}{2} \sum_{i=1}^n\sum_{j=1}^{k} \Big(} + \E_{\vec{\underline{v}}\sim U}\Pr_{\vec{s}\sim (\mathbf{D}^\vec{\underline{v}})^N} [A(\vec{s})_{i,j}) \ne \underline{v}_{i,j}] \Big) \\
    & \ge  \Omega(nk), && \text{(by Lemma~\ref{lem:finite_lb_error_prob})}
\end{align*}
where $d(\cdot,\cdot)$ denotes the hamming distance. By Lemma~\ref{lem:finite_lb_error}, this implies a $\Omega(\epsilon)$ error of the output. Therefore, there exists a distribution $\mathbf{D^v}$ that cannot be learned by $A$ with $O(\epsilon)$ additive error.

\section{Missing Proofs about Prophet Inequality}
\label{app:app-prophet}

\subsection{Optimal Hypothesis}
\label{app:app-prophet-optimal}

An optimal strategy of prophet inequality when $\mathbf{D}$ has bounded support, denote as $S_{\mathbf{D}}$, is called \emph{backward induction}, where we recursively compute the optimal reward for items appear behind $i$ and set the thresholds $ \theta_i$ for item $i$. The algorithm for setting the strategy is as follows:

\begin{algorithm}[t]
	\SetAlgoNoLine
    $\theta_n \leftarrow 0$\;
    $\opt(D_n)=\E_{t_n\sim D_n}[t_n]$\;
	\For{$i$ from $n-1$ to 1}{
	    $ \theta_i \leftarrow \opt(\mathbf{D}_{\ge i+1})$\;
	    $\opt(\mathbf{D}_{\ge i}) = \E_{t_i\ge \theta_i}[t_i] + \Pr[t_i< \theta_i]\opt(\mathbf{D}_{\ge i+1})$\;
	}
	\emph{// online strategy}\\
	$i\leftarrow 1$\;
	\While{$i\le n$}{
	    \eIf{$t_i \ge \theta_i$}{
	        Accept $t_i$ and stop\;
	    }{
	        $i\leftarrow i+1$ \emph{// observe the next reward}
	    }
    	\If{i=n}{
    	    Accept $t_n$ and stop \emph{// if no item has been accepted, accept the last one}
    	}
	}
	\caption{Optimal Strategy for Prophet Inequality in Bounded-support Case}
	\label{alg:backward_induction}	
\end{algorithm}

One particular note for this strategy is that the thresholds for 
the last $n-i+1$ dimension of $\mathbf{D}$ is independent of the arrivals there are before $t_i$. Therefore, in further discussion we can abuse $h_{\ge i}(\mathbf{D})$ to denote the expected reward of running the last $n-i+1$ dimension of an backward induction strategy $S$ on $\mathbf{D}_{\ge i}=\prod_{j=i}^n D_j$.

\subsection{Discretization and Sample Complexity: Proof of Theorem~\ref{thm:finite-prophet}}

Let $\mathbf{D}_{\epsilon/2}$ be the discretized version of $\mathbf{D}$ obtained from rounding the values of each marginal distribution $D_i$ down to the nearest multiples of $\epsilon$.
For all type $\mathbf{t}\sim \mathbf{D}$, define its downward discretization 
\[\mathbf{t}_{\epsilon/2}=\lfloor \frac{2\mathbf{t}}{\epsilon} \rfloor \cdot \frac{\epsilon}{2} ~,\]
also, for the optimal strategy $S_{\mathbf{D}}$ define a coupling optimal strategy  $S_{\mathbf{D}}'$ for $\mathbf{t}_{\epsilon/2}\sim \mathbf{D}_{\epsilon/2}$: 
First re-sample 
\[
    \mathbf{r}\sim \prod_{i=1}^n D_i(t ~|~ t\in [(t_i)_{\epsilon/2}, (t_i)_{\epsilon/2}+\frac{\epsilon}{2})) ~,
\]
then perform the original $S_{\mathbf{D}}$ on $\mathbf{t}'=\mathbf{t}_{\epsilon/2} + \mathbf{r}$ and return the accepted item. We introduce the re-sample step because $\mathbf{t}'$ and $\mathbf{t}\sim \mathbf{D}$ have the same distribution.
Hence at any step $i\in[n]$, the probability that $S_{\mathbf{D}}'$ accepts $\mathbf{t}_{\epsilon/2}\sim \mathbf{D}_{\epsilon/2}$ equals to that of $S_{\mathbf{D}}$ accepts $\mathbf{t}\sim \mathbf{D}$. We further show that the expected reward of the coupling optimal strategy
is an $\epsilon$-additive approximation of the original one:

\begin{lemma}\label{lem:prophet-coupling-optimal-apx}
Under the above definition, we have
\[
    \E_{S_{\mathbf{D}}'}(h_{S_{\mathbf{D}}'}(\mathbf{D}_{\epsilon/2})) \ge h_{S_{\mathbf{D}}}(\mathbf{D}) - \frac{\epsilon}{2} ~.    
\]
\end{lemma}
\begin{proof}
    \begin{align*}
        \E_{S_{\mathbf{D}}'}(h_{S_{\mathbf{D}}'}(\mathbf{D}_{\epsilon/2})) &= \int_{\mathbf{t}\in[0,1]^n} \Pr(\mathbf{t'}=\mathbf{t})\cdot h_{S_{\mathbf{D}}'}(\mathbf{t}) d\mathbf{t}\\
        &\le \int_{\mathbf{t}\in[0,1]^n} f_{\mathbf{D}}(\mathbf{t})\cdot \left\lfloor \frac{2 h_{S_{\mathbf{D}}'}(\mathbf{t}) }{\epsilon}\right\rfloor \cdot \frac{\epsilon}{2} d\mathbf{t}\\
        &\le \int_{\mathbf{t}\in[0,1]^n} f_{\mathbf{D}}(\mathbf{t})\cdot h_{S_{\mathbf{D}}}(\mathbf{t})  d\mathbf{t} - \frac{\epsilon}{2}\\
        &=h_{S_{\mathbf{D}}}(\mathbf{D})- \frac{\epsilon}{2}
    \end{align*}
\end{proof}

Now we go to the proof of Theorem~\ref{thm:finite-prophet}. 
Let $\mathbf{E}_{\epsilon/2}$ be the distribution obtained from rounding down the values of each dimension of $\mathbf{E}$ to the nearest supports of $\frac{\epsilon}{2}$. 
We want to show that, when $N\ge C\cdot\frac{n}{\epsilon^3}\log(\frac{n}{\epsilon\delta})$, $h_{S_{\mathbf{E}_{\epsilon/2}}}$ will become the near-optimal hypothesis of $\mathbf{D}$.

Since $\mathbf{E}_{\epsilon/2}$ is also the empirical distribution of $\mathbf{D}_{\epsilon/2}$, and has finite support with size $\frac{1}{\epsilon/2}$ in each dimension, a corollary of Theorem~\ref{thm:error_finite}
shows that when $N \ge C\cdot\frac{n}{\epsilon^3}\log(\frac{n}{\epsilon\delta})$,
\begin{equation}
     h_{ S_{\mathbf{D}_{\epsilon/2}}} (\mathbf{D}_{\epsilon/2}) - h_{S_{\mathbf{E}_{\epsilon/2}}} (\mathbf{D}_{\epsilon/2} ) \le \frac{\epsilon}{2}  \label{eqn:discretize-aprox-prophet}
\end{equation}
Then
\begin{align*}
     h_{S_{\mathbf{D}}}(\mathbf{D}) &\le \E_{S_{\mathbf{D}}'}(h_{S_{\mathbf{D}}'}(\mathbf{D}_{\epsilon/2})) + \frac{\epsilon}{2} && \text{( Lemma~\ref{lem:prophet-coupling-optimal-apx})} \\
     &\le h_{S_{\mathbf{D}_{\epsilon/2}}}(\mathbf{D}_{\epsilon/2}) + \frac{\epsilon}{2} && \text{(Optimality of $S_{\mathbf{D}_{\epsilon/2}}$)}\\
     &\le h_{S_{\mathbf{E}_{\epsilon/2}}}(\mathbf{D}_{\epsilon/2}) + \epsilon && \text{(From (\ref{eqn:discretize-aprox-prophet}))}
\end{align*}

It remains to show $h_{S_{\mathbf{E}_{\epsilon/2}}}(\mathbf{D}_{\epsilon/2})$ approximates $h_{S_{\mathbf{E}_{\epsilon/2}}}(\mathbf{D})$, i.e. the actual expected reward from the learned hypothesis. We elaborate it as follows:

\begin{lemma}\label{lem:prophet-discretize-empirical-mono}
\[
    h_{S_{\mathbf{E}_{\epsilon/2}}}(\mathbf{D}_{\epsilon/2}) \le h_{S_{\mathbf{E}_{\epsilon/2}}}(\mathbf{D})  ~.
\]
\end{lemma}
\begin{proof}
    Suppose $\mathbf{E}_{\epsilon/2}$ is specified by thresholds $\mathbf{\theta} = (\theta_1,\cdots,\theta_n)$. We can assume without loss of generality that each $\theta_i$ is the multiple of $\frac{\epsilon}{2}$, since rounding the thresholds up does not affect the behavior of the strategy on $\mathbf{E}_{\epsilon/2}$. 

    Therefore, for any type $\mathbf{t}$ and its downward discretization $\mathbf{t}_{\epsilon/2}$, $S_{\mathbf{E}_{\epsilon/2}}$ accepts $\mathbf{t}$ at the $i^{th}$ step if and only if $S_{\mathbf{E}_{\epsilon/2}}$ accepts $\mathbf{t}_{\epsilon/2}$ at the $i^{th}$ step. This gives that

    \begin{align*}
        h_{S_{\mathbf{E}_{\epsilon/2}}}(\mathbf{D}) &= \int_{\mathbf{t}\in[0,1]^n} f_{\mathbf{D}}(\mathbf{t}) h_{S_{\mathbf{E}_{\epsilon/2}}}(\mathbf{t})d\mathbf{t} \\
        &\ge \int_{\mathbf{t}\in[0,1]^n}f_{\mathbf{D}}(\mathbf{t}) h_{S_{\mathbf{E}_{\epsilon/2}}}(\mathbf{t}_{\epsilon/2})d\mathbf{t} \\
        &= \int_{\mathbf{t}\in[0,1]^n}f_{\mathbf{D}_{\epsilon/2}}(\mathbf{t}) h_{S_{\mathbf{E}_{\epsilon/2}}}(\mathbf{t}_{\epsilon/2}) d\mathbf{t}\\
        &= h_{S_{\mathbf{E}_{\epsilon/2}}}(\mathbf{D}_{\epsilon/2}) ~.
    \end{align*}
\end{proof}

With this lemma in hand, we can complete the proof of Theorem~\ref{thm:finite-prophet}.

\subsection{Strong Monotonicity: Proof of Lemma~\ref{lem:strong-monotone-prophet}}
\label{app:strong-monotone-prophet}

From now on, we abuse $h_{\ge i}(\mathbf{D})$ to denote the expected revenue of running the last $n-i+1$ dimension of $S_\mathbf{\tilde{D}_{\ge i}}$ on $\mathbf{D}_{\ge i}$.
We want to show by backward induction on $i$ that $\forall$ product distributions $\mathbf{D}, \tilde{\mathbf{D}}$ such that $\mathbf{D}\succeq \tilde{\mathbf{D}}$, the expected reward of performing $S_{\tilde{\mathbf{D}}}$ on the last $n-i+1$ dimension of $\mathbf{D}$ is at least that of performing $S_{\tilde{\mathbf{D}}}$ on the last $n-i+1$ dimension of $\tilde{\mathbf{D}}$, i.e.
\[
h_{\ge i} ( \mathbf{D}_{\ge i}) \ge h_{\ge i} ( \mathbf{\tilde{D}}_{\ge i})
~.
\]

\textbf{Base case:} When $i=n$, there is only one item with value $t_n$, and $S_\mathbf{\tilde{D}_{\ge n}}$ accepts it with probability $1$ and obtains the reward $t_n$. Therefore, we have
 \begin{align*} 
    h_{\ge n}(\mathbf{D}) &= \E_{D_n}[t_n] = \int_{t=0}^{\infty}q^{D_n}(t) dt\\
    &\ge \int_{t=0}^{\infty}q^{\tilde{D}_n}(t) dt && \text{($D_n \succeq \tilde{D}_n$)}\\
    &= \E_{\tilde{D}_n}[t_n] = h_{\ge n}(\tilde{\mathbf{D}})
    ~.
\end{align*}

\textbf{Inductive step:} Assume the induction hypothesis holds for all $j>i$, i.e. $\forall j>i$ $h_{\ge j}(\mathbf{D}) \ge h_{\ge j}(\tilde{\mathbf{D}})$. Then since $h_{\ge i}(\mathbf{D})$ satisfies the following recursion:
\[
    h_{\ge i}(\mathbf{D})
    = \Pr_{D_i}[t_i \ge \theta_i] \cdot \E_{D_i}[t_i ~|~ t_i \ge \theta_i] + \Pr_{D_i}[t_i < \theta_i]\cdot h_{\ge i+1}(\mathbf{D})
\]
where the first term on the right-hand-side is the expected reward when the $i^{th}$ item is accepted, while the second one is the expected reward when the strategy accepts subsequent item. A similar recursion holds for $h_{\ge i}(\tilde{\mathbf{D}})$:
\[
    h_{\ge i}(\tilde{\mathbf{D}})
    = \Pr_{D_i}[t_i \ge \theta_i] \cdot \E_{D_i}[t_i ~|~ t_i \ge \theta_i] + \Pr_{D_i}[t_i < \theta_i]\cdot h_{\ge i+1}(\mathbf{D})
\]

We then compare the first and the second term of $h_{\ge i}(\mathbf{D})$ and $h_{\ge i}(\tilde{\mathbf{D}})$ respectively. 
\begin{align}
    & h_{\ge i}(\mathbf{D}) \nonumber \\
    = & \Pr_{D_i}[t_i \ge \theta_i] \cdot \E_{D_i}[t_i ~|~ t_i \ge \theta_i] + \Pr_{D_i}[t_i < \theta_i]\cdot h_{\ge i+1}(\mathbf{D}) \nonumber \\
    \ge & \Pr_{D_i}[t_i \ge \theta_i] \cdot \E_{\tilde{D}_i}[t_i ~|~ t_i \ge \theta_i] + (1 - \Pr_{D_i}[t_i \ge \theta_i]) \cdot h_{\ge i+1}(\tilde{\mathbf{D}}) \nonumber \\
    = & \Pr_{\tilde{D}_i}[t_i \ge \theta_i] \cdot \E_{\tilde{D}_i}[t_i ~|~ t_i \ge \theta_i] + (\Pr_{D_i}[t_i \ge \theta_i] -\Pr_{\tilde{D}_i}[t_i \ge \theta_i]) \cdot \E_{\tilde{D}_i}[t_i ~|~ t_i \ge \theta_i] \nonumber \\
    &  + (1 - \Pr_{D_i}[t_i \ge \theta_i]) \cdot h_{\ge i+1}(\tilde{\mathbf{D}})  \label{eqn:bounded_prophet_strongmono}
\end{align}    
where the first inequality comes from $D_i \succeq \tilde{D}_i$ and the induction hypothesis. Furthermore, from the optimality of $S_{\tilde{\mathbf{D}}}$ on $\tilde{\mathbf{D}}$,we must have
\[
\E_{\tilde{D}_i}[t_i ~|~ t_i \ge \theta_i] \ge h_{\ge i+1}(\tilde{\mathbf{D}}) ~,
\]
Otherwise $S_{\tilde{\mathbf{D}}}$ could discard $t_i$ unconditionally and achieve higher expected revenue.    Therefore, 
\begin{align*}
    (\ref{eqn:bounded_prophet_strongmono}) &\ge
    \Pr_{\tilde{D}_i}[t_i \ge \theta_i] \cdot \E_{\tilde{D}_i}[t_i ~|~ t_i \ge \theta_i] + (\Pr_{D_i}[t_i \ge \theta_i] -\Pr_{\tilde{D}_i}[t_i \ge \theta_i]) \cdot h_{\ge i+1}(\tilde{\mathbf{D}})   \\
    &  + (1 - \Pr_{D_i}[t_i \ge \theta_i]) \cdot h_{\ge i+1}(\tilde{\mathbf{D}}) \\
    & = \Pr_{\tilde{D}_i}[t_i \ge \theta_i] \cdot \E_{\tilde{D}_i}[t_i ~|~ t_i \ge \theta_i] + \Pr_{\tilde{D}_i}[t_i < \theta_i] \cdot h_{\ge i+1}(\tilde{\mathbf{D}}) \\
    & = h_{\ge i}(\tilde{\mathbf{D}}) ~.
\end{align*}

\subsection{Prophet Inequality with i.i.d.\ Unbounded Rewards}
\label{app:prophet-unbounded}

In \cite{correa2018prophet}, an $\epsilon$-approximately optimal strategy for known distribution in the unbounded support and i.i.d. case has been introduced. Denote the strategy for distribution $D$ as $R_{\mathbf{D}}$. We restate the algorithm to generate $R_{\mathbf{D}}$ is as follows:

\begin{algorithm}[t]
	\SetAlgoNoLine
    Solve differential equation 
        $y' = y(\log(y)-1)-(\beta-1)$ and $y(0)=1$\\
        where $\beta \approx 1/0.745$ \;
	\For{$i$ from $1$ to $n$}{
	    $\epsilon_i \leftarrow 1-y(i/n)^{1/(n-1)}$\;
	}
	\emph{// online strategy}\;
	$i\leftarrow 1$\;
	\While{$i\le n$}{
	    \If{$\epsilon_i<\frac{\epsilon}{n}$}{
            $\epsilon_i\leftarrow 0$ \emph{// Skip when acceptance probability $<\frac{\epsilon}{n}$}\;
	    }
	    \eIf{$q^{D_i}(t_i) \le \epsilon_i$}{
	        Accept $t_i$ and stop\;
	    }{
	        $i\leftarrow i+1$\;
	    }
	}
	\caption{Approximately Optimal Strategy for Unbounded Support Case}
	\label{alg:apx_for_unbounded_prophet}	
\end{algorithm}

    

In the following discussion, we will show a new sample complexity  bound of achieving $\epsilon$-multiplicative approximation in the unbounded optimal stopping game.

\begin{lemma}\label{lem:ub_unbounded_prophet}
For arbitrary distribution $D$, the sample complexity required for Algorithm \ref{alg:apx_for_unbounded_prophet} is at most $\tilde{O}(\frac{n}{\epsilon^2})$.
\end{lemma}

The algorithm is to run the strategy $R$ on a \emph{dominated empirical distribution} $\tilde{E}$, which is defined below:
\begin{equation*}
	F_{\tilde{E}}(x) = \min \{1,  F_{E}(x)-\sqrt{\frac{2F_{E}(x)(1-F_{E}(x))\ln (2Nn\delta^{-1})}{N}} -\frac{4\ln (2Nn\delta^{-1})}{N} \}
	~.
\end{equation*}
In Lemma 5 of \citet{guo2019settling},  it is shown that with high probability $D \succeq  \tilde{E}$ via a standard concentration bound.

In the following discussion, denote the stopping time of running strategy $R$ on input $\mathbf{t}$ as $\tau(R,\mathbf{t})$

\begin{lemma}
With high probability over samples for the algorithm,
\[
h_{R_{\mathbf{D}}}(\mathbf{D}) - h_{R_{\tilde{\mathbf{E}}}}(\mathbf{D}) < Pr_{\mathbf{t}\sim D^n}[\tau(R_{\tilde{\mathbf{E}}},\mathbf{t}) < \tau(R_{\mathbf{D}},\mathbf{t})] \cdot \opt(\mathbf{D}) ~.
\]
\end{lemma}
\begin{proof}

    Since $D\succeq \tilde{E}$ with high probability, $\forall t\in[n]$ the value threshold in $R_{\tilde{\mathbf{E}}}$ is lower than that of $R_{\mathbf{D}}$, i.e.
    \[
    F_{\tilde{E}}((F_D)^{-1}(1-\epsilon_t))>F_{D}((F_D)^{-1}(1-\epsilon_t))~,
    \]
    Therefore, fix an input value configuration $\mathbf{t}$, $\tau(R_{\tilde{\mathbf{E}}},\mathbf{t}) \le \tau(R_{\mathbf{D}},\mathbf{t})$, and the only case where the revenue obtained from $R_{\tilde{\mathbf{E}}}$ is smaller than that from $R_D$ should be $\tau(R_{\tilde{\mathbf{E}}},\mathbf{t}) < \tau(R_{\mathbf{D}},\mathbf{t})$.
    
    Now it suffices to show that
    \[\E[h_{R_{\mathbf{D}}}(\mathbf{D})-h_{R_{\tilde{\mathbf{E}}}}(\mathbf{D}) ~|~ \tau(R_{\tilde{\mathbf{E}}},\mathbf{t}) < \tau(R_{\mathbf{D}},\mathbf{t})] = O(\opt(\mathbf{D})) ~.\]
    For all $t\in[n]$, define $A_t$ as the set of input such that $R_{\tilde{\mathbf{E}}}$ accepts at time $t$ but $R_{\mathbf{D}}$ does not accept:
    
    Then we can rewrite the above conditioned expected difference of revenue as follows:
    \begin{align*}
        & \E[h_{R_{\mathbf{D}}}(\mathbf{D})-h_{R_{\tilde{\mathbf{E}}}}(\mathbf{D}) ~|~ \tau(R_{\tilde{\mathbf{E}}},\mathbf{t}) < \tau(R_{\mathbf{D}},\mathbf{t})]\\
        =~ &  \E[h_{R_{\mathbf{D}}}(\mathbf{D})-h_{R_{\tilde{\mathbf{E}}}}(\mathbf{D}) ~|~ \mathbf{t}\in \cup_{t=1}^{n-1}A_t]\\
        \le~ & \E[h_{R_{\mathbf{D}}}(\mathbf{D}) ~|~ \mathbf{t}\in \cup_{t=1}^{n-1}A_t] && (h_{R_{\tilde{\mathbf{E}}}}(\mathbf{D})>0)\\
        =~ & \max_{t\in[n-1]} \E[h_{R_{\mathbf{D}}}(\mathbf{D}) ~|~ \mathbf{t}\in A_t]\\
        =~ & \max_{t\in[n-1]} \E[h_{R_{\mathbf{D}}}(\mathbf{D}) ~|~ R_D \text{ does not accept before }t]\\
        \le~ & \opt(\mathbf{D})
    \end{align*}
\end{proof}

\begin{lemma}
When $m\ge \tilde{O}(n\epsilon^{-2})$, with high probability over samples for the algorithm,
\[
Pr_{\mathbf{t}\sim D^n}[\tau(R_{\tilde{\mathbf{E}}},\mathbf{t}) < \tau(R_{\mathbf{D}},\mathbf{t})] < O(\epsilon) ~.
\]
\end{lemma}
\begin{proof}

\begin{align}
    &Pr_{\mathbf{t}\sim D^n}[\tau(R_{\tilde{\mathbf{E}}},\mathbf{t}) < \tau(R_{\mathbf{D}},\mathbf{t})] \\
    \le &
    Pr_{\mathbf{t}\sim D^n}[\mathbf{t} \in \sum_{t=1}^{n-1}A_t] \nonumber \\
    \le & \sum_{t=1}^{n-1} Pr_{\mathbf{t}\sim D^n}[\epsilon_t\le q^{\tilde{E}}(X_i)\le \epsilon_t+\sqrt{\frac{8\epsilon_t(1-\epsilon_t)\ln (2Nn\delta^{-1})}{N}}+\frac{7\ln (2Nn\delta^{-1})}{N}] \nonumber \\
    \le & \sum_{t=1}^{n-1}(\sqrt{\frac{8\epsilon_t(1-\epsilon_t)\ln (2Nn\delta^{-1})}{N}}+\frac{7\ln (2Nn\delta^{-1})}{N})\\
     = & \sum_{t=1}^{n-1}\sqrt{\frac{8\epsilon_t(1-\epsilon_t)\ln (2Nn\delta^{-1})}{N}}+ O(\epsilon^2)
    ~,\label{eqn:unbounded_prophet_different_stoptime_ub}
\end{align}
the second inequality is also shown in Lemma 7 of \citet{guo2019settling} to be hold with high probability.
Recall from Algorithm~\ref{alg:apx_for_unbounded_prophet} that $\epsilon_t=1-y(\frac{t}{n})^{1/(n-1)}$. Now we bound (\ref{eqn:unbounded_prophet_different_stoptime_ub}) for $y(\frac{t}{n})<\frac{1}{n}$ or $y(\frac{t}{n})>\frac{1}{n}$:

\textbf{Case 1: $y(\frac{t}{n})\ge\frac{1}{n}$. }In this case, 
\[
\epsilon_t = 1-y^{\frac{1}{n-1}} \le 1- e^{\frac{\log y}{n-1}} \le 1- e^{-\frac{\log n}{n-1}}\le \frac{\log n}{n} ~,
\]
Therefore when $m\ge n\epsilon^{-2}\log n$
\[
\sum_{y(\frac{t}{n})\ge\frac{1}{n}}\sqrt{\frac{\epsilon_t(1-\epsilon_t)\ln (2Nn\delta^{-1})}{N}} \le n\cdot \sqrt{\frac{\log n/n \cdot 1\cdot \ln (2Nn\delta^{-1})}{N}} = O(\epsilon) ~.
\]

\textbf{Case 2: $y(\frac{t}{n}) < \frac{1}{n}$. } Since $y(x)\in [0,1]$ when $x\in[0,1]$, we have $\forall x\in[0,1]$,
\[
y'(x) = y(\log y-1)-(\beta-1) \le -(\beta - 1) \le -0.3414
\]
Therefore 
\[
|\{t\in[n-1] \text{, s.t. } y(\frac{t}{n}) < \frac{1}{n}\} | \le \frac{1}{n} \cdot \frac{1}{0.3414} \cdot n \le 3 ~,
\]
and when $N\ge n\epsilon^{-2}\log n$,
\[
\sum_{y(\frac{t}{n})<\frac{1}{n}}\sqrt{\frac{\epsilon_t(1-\epsilon_t)\ln (2Nn\delta^{-1})}{N}} \le 3\cdot \sqrt{\frac{\epsilon_t(1-\epsilon_t)}{N}\ln (2Nn\delta^{-1})}  =O (\frac{\epsilon}{\sqrt{n}}) ~.
\]

Combining the two cases, we have
\begin{align*}
    &Pr_{\mathbf{t}\sim D^n}[\tau(R_{\tilde{\mathbf{E}}},\mathbf{t}) < \tau(R_{\mathbf{D}},\mathbf{t})]\\
    &\le 4\sum_{y(\frac{t}{n})<\frac{1}{n}}\sqrt{\frac{\epsilon_t(1-\epsilon_t)\ln (2Nn\delta^{-1})}{N}} + 4\sum_{y(\frac{t}{n})\ge \frac{1}{n}}\sqrt{\frac{\epsilon_t(1-\epsilon_t)\ln (2Nn\delta^{-1})}{N}} +O(\epsilon^2)\\
    &= O(\epsilon) ~.
\end{align*}

\end{proof}

\section{Missing Proofs about Pandora's Problem}
\label{app:pandora}

\subsection{Optimal Hypothesis}
\label{app:pandora-optimal}

\paragraph{Optimal strategy of Pandora's problem}
An optimal strategy $S_{\mathbf{D}}$ for $\mathbf{D}$, introduced by \cite{weitzman1979optimal}, opens the boxes sequentially according to its \emph{reservation value} $\sigma_i$, the threshold of the maximum realized values below which opening the $i+1^{th}$ box will give rise to a higher expected reward.
A formal definition of $\sigma_i$ is as follows:
\begin{equation*}
    \sigma_i \defeq \inf_{\sigma}(\E_{t_i\sim D_i}[(t_i - \sigma)^+]=c_i)
    ~.
\end{equation*}
we assume without loss of generality that the reserve value is non-increasing with the index of each box, i.e. $\sigma_1\ge \sigma_2 \ge \sigma_n$. 

Also, for convenience we use $U_i$ to denote the maximum value among the first $i$ boxes:
    \begin{equation*}
        U_i\defeq\max_{j\le i}t_j
        ~,
    \end{equation*}
and let $U_0\defeq 0$. 

We restate the optimal strategy $S_{\mathbf{D}}$ in \cite{weitzman1979optimal} as Algorithm~\ref{alg:one}:
\begin{algorithm}[t]
	\SetAlgoNoLine
	$i\leftarrow 1$\;
	\While{$i\le n$}{
	    Open the $i^{th}$ box and set $U_i \leftarrow \max_{j\le i}t_j$\;
	    \eIf{$U_i \ge \sigma_i$}{
	         Accept $U_i$ and stop  \emph{// accept the highest opened box so far}\;
	    }{
	        $i\leftarrow i+1$\;
	    }
	}
	\If{i=n}{
	    Accept $U_n$ and stop \emph{// if no item has been accepted, accept the box with highest reward}\;
	}
	\caption{Frequency Number Computation}
	\label{alg:one}
\end{algorithm}


\subsection{Discretization and Sample Complexity: Proof of Theorem~\ref{thm:finite-pandora}}

The proof of Theorem~\ref{thm:finite-pandora} is almost the same as that of Theorem~\ref{thm:finite-prophet}. We include it here only for completeness.

Let $\mathbf{D}_{\epsilon/2}$ be the discretized version of $\mathbf{D}$ obtained from rounding the values of each marginal distribution $D_i$ down to the nearest multiples of $\epsilon$. 
Also, for all type $\mathbf{t}\sim \mathbf{D}$, let $\mathbf{t}_{\epsilon/2}$ be its downward discretization to the multiples of $\frac{\epsilon}{2}$.
For the optimal strategy $S_{\mathbf{D}}$ define a coupling optimal strategy  $S_{\mathbf{D}}'$ for discretized type $\mathbf{t}_{\epsilon}/2$: 
First re-sample 
\[
    \mathbf{r}\sim \prod_{i=1}^n D_i(t ~|~ t\in [(t_i)_{\epsilon/2}, (t_i)_{\epsilon/2}+\frac{\epsilon}{2})) ~,
\]
then perform the original $S_{\mathbf{D}}$ on $\mathbf{t}'=\mathbf{t}_{\epsilon/2} + \mathbf{r}$ and return the accepted item. It is easy to see that after the re-sample step $\mathbf{t}'$ has the same distribution as $\mathbf{t}\sim \mathbf{D}$.
Hence at any step $i\in[n]$ and for any $j\in[i]$, the probability that $S_{\mathbf{D}}'$ accepts the $j^{th}$ box of $\mathbf{t}_{\epsilon/2}\sim \mathbf{D}_{\epsilon/2}$ equals to that of $S_{\mathbf{D}}$ accepts the $j^{th}$ box of $\mathbf{t}\sim \mathbf{D}$. We further show that the expected reward of the coupling optimal strategy
is an $\epsilon$-additive approximation of the original one:

\begin{lemma}\label{lem:pandora-coupling-optimal-apx}
Under the above definition, we have
\[
    \E_{S_{\mathbf{D}}'}(h_{S_{\mathbf{D}}'}(\mathbf{D}_{\epsilon/2})) \ge h_{S_{\mathbf{D}}}(\mathbf{D}) - \frac{\epsilon}{2} ~.    
\]
\end{lemma}
\begin{proof}
Same as the proof of Lemma~\ref{lem:prophet-coupling-optimal-apx}.
\end{proof}

Now we go to the proof of Theorem~\ref{thm:finite-pandora}. 
Let $\mathbf{E}_{\epsilon/2}$ be the distribution obtained from rounding down the values of each dimension of $\mathbf{E}$ to the nearest supports of $\frac{\epsilon}{2}$. 
We want to show that, when $N\ge C\cdot\frac{n^3}{\epsilon^3}\log(\frac{n}{\epsilon\delta})$, $h_{S_{\mathbf{E}_{\epsilon/2}}}$ will become the near-optimal hypothesis of $\mathbf{D}$.

Since $\mathbf{E}_{\epsilon/2}$ is also the empirical distribution of $\mathbf{D}_{\epsilon/2}$, and has finite support with size $\frac{1}{\epsilon/2}$ in each dimension. Because the value value is bounded in $[-n,1]$, a corollary of Theorem~\ref{thm:error_finite}
shows that when $N \ge C\cdot\frac{n^3}{\epsilon^3}\log(\frac{n}{\epsilon\delta})$ for a large enough constant $C$,
\begin{equation}
     \frac{h_{ S_{\mathbf{D}_{\epsilon/2}}} (\mathbf{D}_{\epsilon/2})+n}{n+1} - \frac{h_{S_{\mathbf{E}_{\epsilon/2}}} (\mathbf{D}_{\epsilon/2} )+n}{n+1} \le \frac{\epsilon}{2(n+1)}  \label{eqn:discretize-aprox-pandora}
\end{equation}
Then
\begin{align*}
     h_{S_{\mathbf{D}}}(\mathbf{D}) &\le \E_{S_{\mathbf{D}}'}(h_{S_{\mathbf{D}}'}(\mathbf{D}_{\epsilon/2})) + \frac{\epsilon}{2} && \text{( Lemma~\ref{lem:pandora-coupling-optimal-apx})} \\
     &\le h_{S_{\mathbf{D}_{\epsilon/2}}}(\mathbf{D}_{\epsilon/2}) + \frac{\epsilon}{2} && \text{(Optimality of $S_{\mathbf{D}_{\epsilon/2}}$)}\\
     &\le h_{S_{\mathbf{E}_{\epsilon/2}}}(\mathbf{D}_{\epsilon/2}) + \epsilon && \text{( (\ref{eqn:discretize-aprox-pandora}))}
\end{align*}

It remains to show $h_{S_{\mathbf{E}_{\epsilon/2}}}(\mathbf{D}_{\epsilon/2})$ approximates $h_{S_{\mathbf{E}_{\epsilon/2}}}(\mathbf{D})$, i.e. the actual expected reward from the learned hypothesis. We elaborate it as follows:

\begin{lemma}
\[
    h_{S_{\mathbf{E}_{\epsilon/2}}}(\mathbf{D}_{\epsilon/2}) \le h_{S_{\mathbf{E}_{\epsilon/2}}}(\mathbf{D})  ~.
\]
\end{lemma}
\begin{proof}
Same as the proof of Lemma~\ref{lem:prophet-discretize-empirical-mono}.
\end{proof}

With this lemma in hand, we can complete the proof of Theorem~\ref{thm:finite-pandora}.

\subsection{Strong Monotonicity: Proof of Lemma~\ref{lem:strong-monotone-pandora}}
\label{app:strong-monotone-pandora}


It suffices to show that for any $\mathbf{D} \succeq \mathbf{\tilde{D}}$,
\[
    h_{S_\mathbf{\tilde{D}}}(\mathbf{D}) \ge h_{S_\mathbf{\tilde{D}}}(\mathbf{\tilde{D}}) ~.
\]


\begin{lemma}\label{lem:pandora_optimality}
    For any $H\le  \sigma_{i+1}$,
    \[
        h_{S_{\tilde{\mathbf{D}}}}(\tilde{\mathbf{D}} | U_i=H) \le \sigma_{i+1}-\sum_{j=1}^i c_j
    \]
\end{lemma}
\begin{proof}
    Since $S_{\Tilde{\mathbf{D}}}$ is the optimal strategy, $h_{S_{\tilde{\mathbf{D}}}}(\tilde{\mathbf{D}} | U_i=H)$ is monotone in $H$, so it is only necessary to show the case when $U_i=H= \sigma_{i+1}$. But in this case, simply choosing the largest among first $i$ boxes would give a revenue of $\sigma_{i+1}-\sum_{j=1}^i c_j$, so 
    \[
        h_{S_{\tilde{\mathbf{D}}}}(\tilde{\mathbf{D}} | U_i=\sigma_{i+1})\le \sigma_{i+1}-\sum_{j=1}^i c_j 
    \]
    is true due to the optimality of the mechanism.
\end{proof}

We will use backward induction from $n$ to $0$ to prove the following statement.
\begin{lemma}\label{lem:pandora_induction}
    For any $0\le i\le n$ and any $u_i'\le u_i\le \sigma_{i+1}$,
    \[
    h_{S_{\tilde{\mathbf{D}}}}(\tilde{\mathbf{D}} | U_i=u_i) > h_{S_{\tilde{\mathbf{D}}}}(\tilde{\mathbf{D}} | U_i=u_i')
    ~.
    \]
\end{lemma}
\begin{proof}
    This holds trivially when $i=n$. In the following discussion, assume $i<n$ and the lemma holds for $i+1$.

    If $U_i\le \sigma_{i+1}$, the mechanism will choose to open the next box. 
    In this case, because $D_{i+1}\succeq \tilde{D_{i+1}}$, it suffices to show that for any $t_{i+1}\ge t_{i+1}'$, 
    \begin{equation*}
        h_{S_{\tilde{\mathbf{D}}}}(\mathbf{D} | U_i=u_i,X_{i+1}=t_{i+1}) \ge h_{S_{\tilde{\mathbf{D}}}}(\tilde{\mathbf{D}} | U_i=u_i',X_{i+1}=t_{i+1}')
        ~.
    \end{equation*}
    So it is enough to show that for any $u_{i+1}\ge u_{i+1}'$,
    \begin{equation*}
        h_{S_{\tilde{\mathbf{D}}}}(\mathbf{D} |U_{i+1}=u_{i+1}) \ge h_{S_{\tilde{\mathbf{D}}}}(\tilde{\mathbf{D}} |U_{i+1}=u_{i+1}')
        ~.
    \end{equation*}
    
    We will consider three cases. In the first case, $u_{i+1}\ge u_{i+1}'> \sigma_{i+2}$. Then 
    \[ h_{S_{\Tilde{\mathbf{D}}}}(\mathbf{D})=u_{i+1}-\sum_{j=1}^i c_i\ge u_{i+1}'-\sum_{j=1}^i c_j=h_{S_{\Tilde{\mathbf{D}}}}(\Tilde{\mathbf{D}}) ~.\]
    
    In the second case, $u_{i+1}>\sigma_{i+2}\ge u_{i+1}'$, then we can apply Lemma~\ref{lem:pandora_optimality} and have
    \[
    h_{S_{\Tilde{\mathbf{D}}}}(\Tilde{\mathbf{D}})\le \sigma_{i+2}-\sum_{j=1}^{i+1} c_i< u_{i+1} -\sum_{j=1}^{i+1}c_k = h_{S_{\Tilde{\mathbf{D}}}}(\mathbf{D})
    \]
    
    In third case, $\sigma_{i+2}\ge u_{i+1}\ge u_{i+1}'$, the inequality follows directly from induction assumption.
\end{proof}

\subsection{Tight Bounds: Proof of Theorem~\ref{thm:pandora}}
\label{app:pandora-improved}

\subsubsection{Upper Bound}

We start by recalling the main obstacle for getting an $\tilde{O} \big( \frac{n}{\epsilon^2} \big)$ sample complexity upper bound as a direct corollary of Theorem~\ref{thm:strong-monotone-upper-bound}.
In Pandora's problem, an algorithm may pay a cost up to $1$ to open each box and thus, the range of the realized objective is $[-n, 1]$ instead of $[0, 1]$.
In the main text, we use a simple hypothesis class $\hypoclass$, which has a hypothesis $\hypo_A$ for each algorithm $A$, normalizing its value to be in $[0, 1]$ by letting it be the realized objective of $A$ plus $n$ and scaled by $\frac{1}{n+1}$.
Therefore, to get an $\epsilon$-additive approximation in Pandora's problem, we need a $\frac{\epsilon}{n+1}$-additive approximation w.r.t.\ the general learning problem $\hypoclass$.
Therefore, applying Theorem~\ref{thm:strong-monotone-upper-bound} to this hypothesis class $\hypoclass$ gives only an $\tilde{O} \big( \frac{n^3}{\epsilon^2} \big)$ sample complexity bound.

Although the objective could be as small as $-n$ in the worst cases, intuitively the chance of getting such a bad objective shall be negligible if the algorithm is reasonable w.r.t.\ the underlying distribution.
Indeed, we will reason that it is without loss of generality to consider algorithms that stop whenever the cost exceeds $\log \frac{1}{\epsilon}$.
As a result, we avoid scaling the value of the hypotheses by a $n+1$ factor.

In particular, we consider the following notion of rational algorithms w.r.t.\ a given distribution.

\begin{definition}[Rational Algorithms]
    For any distribution $\vec{D}$ and any cost vector $\vec{c}$, an algorithm $A$ for the Pandora's problem is \emph{rational} w.r.t.\ $\vec{D}$ and $\vec{c}$ if whenever $A$ opens a box $i$, the expected increase in the best observed reward is greater than or equal to the cost $c_i$.
\end{definition}

The next lemma follows by the definition of the optimal algorithm.

\begin{lemma}
    \label{lem:pandora-rational-alg}
    Suppose $A$ is the optimal algorithm w.r.t.\ a distribution $\vec{\tilde{D}}$ and a cost vector $\vec{c}$.
    Then, $A$ is rational w.r.t.\ any distribution $\vec{D}$ that stochastically dominates (including $\vec{\tilde{D}}$ itself), and $\vec{c}$.
\end{lemma}

We will need a standard Bernstein type concentration bound for submartingales.
\begin{lemma}
    \label{lem:bernstein-submartingale}
    Let $S_0,S_1\cdots,S_n$ be a submartingale with respect to filtration $\mathcal{F}_0,\mathcal{F}_1,\cdots,\mathcal{F}_k$. Suppose $S_0=0$, $|S_i-S_{i-1}|\le M$, $\sum_{i=1}^n\E[(S_i-S_{i-1})^2|\mathcal{F}_{i-1}]\le L$, then for any positive $\Delta$,
    \[
        \Pr[S_n< -\Delta]\le \exp\left(\frac{\Delta^2}{2L+(2/3)M\Delta}\right)
        ~.
    \]
\end{lemma}

\begin{lemma}
    \label{lem:pandora-large-cost-bound}
    Suppose an algorithm $A$ is rational w.r.t.\ a distribution $\vec{D}$ and a cost vector $\vec{c}$.
    Then, the probability that $A$ pays a cost more than $\Omega \big( \log \frac{1}{\epsilon} \big)$ is at most $\epsilon$.
\end{lemma}

\begin{proof}
    For $1 \le i \le n$, let $X_i$ be objective after round $i$;
    if the algorithm stops before round $i$, let $X_i = X_{i-1}$.
    Let $X_0 = 0$.
    Then, by that $A$ is rational, we have:
    \[
        \E \big[ X_i | X_1, X_2, \dots, X_{i-1} \big] \ge X_{i-1}
        ~.
    \]
    
    That is, $X_i$'s form a discrete-time submartingale.
    
    We have $-1 \le X_i - X_{i-1} \le 1$ by definition.
    Further, for any round $1 \le i \le n$, $X_i - X_{i-1}$ is upper bounded by the increment in the best observed reward in the round. 
    Therefore, $\sum_{i : X_i \ge X_{i-1}} (X_i - X_{i-1})$ is at most the best observed reward at the end, which is upper bounded by $1$.
    Hence, we have:
    \begin{align*}
        \sum_{i=1}^n \E \big[ (X_i - X_{i-1})^2 \big]
        & 
        \le \sum_{i=1}^n \E \big[ \big| X_i - X_{i-1} \big| \big]
        && \text{($-1 \le X_i - X_{i-1} \le 1$)} \\
        &
        = \E \bigg[ \sum_{i=1}^n (X_{i-1} - X_i) + 2 \sum_{i : X_i \ge X_{i-1}} (X_i - X_{i-1}) \bigg] \\
        &
        = - \E \big[ X_n \big] + 2 \cdot \E \bigg[ \sum_{i : X_i \ge X_{i-1}} (X_i - X_{i-1}) \bigg] \\
        &
        \le 2 
        ~.
        && \text{($\textstyle \sum_{i : X_i \ge X_{i-1}} (X_i - X_{i-1}) \le 1$)}
    \end{align*}
    
    Since the cost is at most $-X_n$ by definition, it suffices to upper bound the probability that $X_n \le - \Omega \big( \log \frac{1}{\epsilon} \big)$.
    Then, the lemma follows by Bernstein's inequality for submartingales.
\end{proof}

In the following arguments, consider a hypothesis class $\hypoclass$, which has a hypothesis for any algorithm $A$ such that its value equals the objective of $A$, \emph{without scaling}.
Further, fixed the cost vector $\vec{c}$, let $\hypo_{\vec{D}}$ be the hypothesis that corresponds to the optimal algorithm for $\vec{D}$ and $\vec{c}$.
Finally, let $\bar{\hypo}_{\vec{D}}$ be the hypothesis that corresponds to a truncated version of the optimal algorithm for $\vec{D}$, which stops whenever the cost exceeds $\Omega \big( \log \frac{1}{\epsilon} \big)$.

\begin{lemma}
    \label{lem:pandora-truncation}
    Fixed any cost vector $\vec{c}$.
    For any $\vec{D} \succeq \vec{\tilde{D}}$, the truncated version of optimal algorithm w.r.t.\ $\vec{\tilde{D}}$ gets an expected value greater than or equal to that of the untruncated version minus $\epsilon$:
    \[
        \bar{\hypo}_{\vec{\tilde{D}}} \big( \vec{D} \big) \ge \hypo_{\vec{\tilde{D}}} \big( \vec{D} \big) - \epsilon
        ~.
    \]
\end{lemma}

\begin{proof}
    By Lemma~\ref{lem:pandora-rational-alg}, $\hypo_{\vec{\tilde{D}}}$ is rational w.r.t.\ $\vec{D}$ and $\vec{c}$.
    Hence, by Lemma~\ref{lem:pandora-large-cost-bound}, the probability that the truncated version $\bar{\hypo}_{\vec{\tilde{D}}}$ and the original version $\hypo_{\vec{\tilde{D}}}$ give different outcomes is at most $\epsilon$.
    Finally, whenever they are different, $\hypo_{\vec{\tilde{D}}}$ gets at most $1$ extra reward in subsequent rounds.
    Putting together proves the lemma.
\end{proof}

We now prove the stated sample complexity upper bound.

\begin{proof}[Proof of Theorem~\ref{thm:pandora} (Upper Bound)]
    We show that the truncated version of PERM gets the stated sample complexity bound.
    We prove an $O(\epsilon)$-additive approximation with the understanding that changing $\epsilon$ by a constant factor does not affect the stated sample complexity bound asymptotically.
    
    It follows from a sequence of inequalities below, similar to those in Section~\ref{sec:strong-monotone}:
    \begin{align*}
        \bar{\hypo}_{\vec{E}} \big( \vec{D} \big)
        &
        \ge \bar{\hypo}_{\vec{E}} \big( \vec{\hat{D}} \big) - O \big( \log \tfrac{1}{\epsilon} \big) \cdot \delta \big( \vec{\hat{D}}, \vec{D} \big)
        && \text{($\bar{\hypo}_{\vec{E}}$ bounded in $[-O \big( \log \tfrac{1}{\epsilon} \big), 1]$)} \\
        &
        \ge \hypo_{\vec{E}} \big( \vec{\hat{D}} \big) - \epsilon - O \big( \log \tfrac{1}{\epsilon} \big) \cdot \delta \big( \vec{\hat{D}}, \vec{D} \big) 
        && \text{(Lemma~\ref{lem:pandora-truncation})} \\
        & 
        \ge \hypo_{\vec{E}} \big( \vec{E} \big) - \epsilon - O \big( \log \tfrac{1}{\epsilon} \big) \cdot \delta \big( \vec{\hat{D}}, \vec{D} \big) 
        && \text{(strong monotonicity, $\vec{\hat{D}} \succeq \vec{E}$)} \\
        & 
        = \opt \big( \vec{E} \big) - \epsilon - O \big( \log \tfrac{1}{\epsilon} \big) \cdot \delta \big( \vec{\hat{D}}, \vec{D} \big) 
        && \text{(definition of $\opt\big(\vec{E}\big)$)} \\
        &
        \ge \opt \big( \vec{\check{D}} \big) - \epsilon - O \big( \log \tfrac{1}{\epsilon} \big) \cdot \delta \big( \vec{\hat{D}}, \vec{D} \big) 
        && \text{(weak monotonicity, $\vec{E} \succeq \vec{\check{D}}$)} \\
        & 
        \ge \bar{\hypo}_{\vec{D}} \big( \vec{\check{D}} \big) - \epsilon - O \big( \log \tfrac{1}{\epsilon} \big) \cdot \delta \big( \vec{\hat{D}}, \vec{D} \big) 
        && \textrm{(definition of $\opt\big(\vec{\check{D}}\big)$)} \\
        & 
        \ge \bar{\hypo}_{\vec{D}} \big( \vec{D} \big) - \epsilon - O \big( \log \tfrac{1}{\epsilon} \big) \cdot \big( \delta \big( \vec{\hat{D}}, \vec{D} \big) + \delta \big( \vec{\check{D}}, \vec{D} \big) \big)
        && \text{($\bar{\hypo}_{\vec{D}}$ bounded in $[- O \big( \log \tfrac{1}{\epsilon} \big), 1]$)} \\
        &
        \ge \hypo_{\vec{D}} \big( \vec{D} \big) - 2 \epsilon - O \big( \log \tfrac{1}{\epsilon} \big) \cdot \big( \delta \big( \vec{\hat{D}}, \vec{D} \big) + \delta \big( \vec{\check{D}}, \vec{D} \big) \big)
        && \text{(Lemma~\ref{lem:pandora-truncation})} \\
        &
        = \opt \big( \vec{D} \big) - 2 \epsilon - O \big( \log \tfrac{1}{\epsilon} \big) \cdot \big( \delta \big( \vec{\hat{D}}, \vec{D} \big) + \delta \big( \vec{\check{D}}, \vec{D} \big) \big) 
        ~.
        &&
        \textrm{(definition of $\opt\big(\vec{D}\big)$)}
    \end{align*}
    
    By Lemma~\ref{lem:hellinger-and-total-variation}, Lemma~\ref{lem:strong-monotone-Dhat-and-D} and Lemma~\ref{lem:strong-monotone-Dcheck-and-D}, we get an $O(\epsilon)$-additive approximation.
\end{proof}

\subsubsection{Lower Bound}

Consider $n$ boxes with cost $\frac{1}{n}$ each.
Consider $2^n$ potential instances, in which the reward distribution of each box is either $D^+$ or $D^-$, defined by the following probability mass functions respectively:
\[
    f_{D^+}(x) = \begin{cases}
        \frac{1+\epsilon}{n} & x = 1 ~; \\
        1 - \frac{1+\epsilon}{n} & x = 0 ~.
    \end{cases}
    \quad
    f_{D^-}(x) = \begin{cases}
        \frac{1-\epsilon}{n} & x = 1 ~; \\
        1 - \frac{1-\epsilon}{n} & x = 0 ~.
    \end{cases}
\]

We will refer to each of these $2^n$ instances by the distribution $\vec{D}$, since the cost vector $\vec{c}$ is fixed.

The next lemma follows by the above definition and and simple calculations which we omit.

\begin{lemma}
    \label{lem:pandora-lb-hellinger}
    The squared Hellinger distance between $D^+$ and $D^-$ is bounded by:
    \[
        \Hellinger^2 \big( D^+, D^- \big) = O \bigg( \frac{n}{\epsilon^2} \bigg)
        ~.
    \]
\end{lemma}

To distinguish the algorithm for Pandora's problem and the learning algorithm, we will refer to the former as a hypothesis.

Since the rewards are either $0$ or $1$, any hypothesis is characterized by an ordered subsequence $i_1, i_2, \dots, i_k$ of the boxes such that it opens the boxes one by one until it gets a reward $1$;
if all $k$ rewards are $0$, it stops and leaves the remaining $n - k$ boxes unopened. 
The optimal hypothesis chooses a box into the subsequence if and only if its distribution equals $D^+$ (order is irrelevant since they are identical). 
Therefore, for any instance $\vec{D}$ defined above, any hypothesis $\hypo$, and any box $1 \le i \le n$, we say that $\hypo$ makes a mistake on box $i$ w.r.t.\ $\vec{D}$ if either $D_i = D^+$ but $i$ isn't in the subsequence chosen by $\hypo$, or $D_i = D^-$ but $i$ is in the subsequence.
We simply say that the algorithm makes a mistake on box $i$ w.r.t\ $\vec{D}$ if it selects a hypothesis that makes such a mistake.
Whether a given learning algorithm makes a mistake might be a random event if it is randomized.

In the rest of the argument, we first argue that the additive approximation error scales linearly with number of mistakes made by the chosen hypothesis.
Then, we argue through a sequence of lemmas that for any algorithm that takes less than $c \cdot \frac{n}{\epsilon^2}$ samples for some sufficiently small constant $c > 0$, there is an instance $\vec{D}$ for which it picks a hypothesis that makes $\Omega \big( n \big)$ mistakes with at least constant probability.

\begin{lemma}
    \label{lem:pandora-lb-mistakes-and-error}
    For any instance $\vec{D}$, if a hypothesis $\hypo$ makes $k$ mistakes, then we have:
    \[
        \hypo \big( \vec{D} \big) \le \opt \big( \vec{D} \big) - \Omega \bigg( \frac{k \epsilon}{n} \bigg)
        ~.
    \]
\end{lemma}

\begin{proof}
    Suppose the instance have $n^+$ and $n^-$ boxes with reward distributions equal to $D^+$ and $D^-$ respectively.
    Further suppose $\hypo$ makes $k^+$ and $k^-$ mistakes on the two types of boxes.
    Hence, $\hypo$ includes $n^+ - k^+$ boxes with distributions equal to $D^+$ and $k^-$ boxes with distributions equal to $D^-$ in its subsequence.
    
    The expected reward minus cost for opening a box with distribution $D^+$ is $\frac{\epsilon}{n}$; 
    opening a box with distribution $D^-$ gives $- \frac{\epsilon}{n}$.
    Further, the probability of opening the $i$-th box in the sequence is equal to the probability that the first $i-1$ rewards are all $0$.
    
    Hence, the optimal is:
    \[
        \opt \big( \vec{D} \big) = \frac{\epsilon}{n} \bigg( 1 + \bigg(1 - \frac{1+\epsilon}{n} \bigg) + \dots + \bigg(1 - \frac{1+\epsilon}{n} \bigg)^{n^+-1} \bigg) \\
        ~.
    \]
    
    The expected value of the hypothesis is at most (when it opens the $n^+ - k^+$ boxes with reward distributions equal to $D^+$ first):
    \[
        \hypo \big( \vec{D} \big) \le \frac{\epsilon}{n} \bigg( \sum_{i=1}^{n^+ - k^+} \bigg(1 - \frac{1+\epsilon}{n} \bigg)^{i-1} - \bigg(1 - \frac{1+\epsilon}{n} \bigg)^{n^+ - k^+} \sum_{i=1}^{k^-} \bigg(1 - \frac{1-\epsilon}{n} \bigg)^{i-1} \bigg)
    \]
    
    Therefore, we have:
    \begin{align*}
        \opt \big( \vec{D} \big) - \hypo \big( \vec{D} \big) 
        &
        \ge \frac{\epsilon}{n} \bigg(1 - \frac{1+\epsilon}{n} \bigg)^{n^+ - k^+} \bigg( \sum_{i=1}^{k^+} \bigg(1 - \frac{1+\epsilon}{n} \bigg)^{i-1} + \sum_{i=1}^{k^-} \bigg(1 - \frac{1-\epsilon}{n} \bigg)^{i-1} \bigg) \\
        &
        \ge \frac{\epsilon}{n} \bigg( \sum_{i=1}^{k^+} \bigg(1 - \frac{1+\epsilon}{n} \bigg)^n + \sum_{i=1}^{k^-} \bigg(1 - \frac{1+\epsilon}{n} \bigg)^n \bigg) \\[1ex]
        &
        = \frac{\epsilon k}{n} \bigg(1 - \frac{1+\epsilon}{n} \bigg)^n \\[2ex]
        &
        \ge \frac{\epsilon k}{n} \exp \big( - 2 - 2\epsilon \big) 
        ~.
    \end{align*}
    
    The last inequality is due to $1-x>e^{-2x}$ for $0 < x < \frac{1}{2}$.
\end{proof}

\begin{lemma}
    \label{lem:pandora-lb-local-mistake-bound}
    For any algorithm $A$, any box $1 \le i \le n$, and any two neighboring instances $\vec{D^+}$ and $\vec{D^-}$ that differ only in the $i$-th coordinate, we have:
    \[
        \Pr \big[ \text{$A$ makes a mistake on box $i$ w.r.t.\ $\vec{D^+}$} \big] + \Pr \big[ \text{$A$ makes a mistake on box $i$ w.r.t.\ $\vec{D^-}$} \big] \ge \Omega(1)
        ~.
    \]
\end{lemma}

\begin{proof}
    By definition, any hypothesis $h$ makes a mistake on box $i$ w.r.t.\ either $\vec{D^+}$ or $\vec{D^-}$.
    Let $\hypoclass^+$ and $\hypoclass^-$ denote the two subsets of hypotheses respectively.
    On the one hand, we have:
    \[
         \Pr \big[ \text{$A$ picks $\hypo \in \hypoclass^+$ given samples from $\vec{D^+}$} \big] + \Pr \big[ \text{$A$ picks $\hypo \in \hypoclass^-$ given samples from $\vec{D^+}$} \big] = 1
         ~.
    \]
    
    On the other hand, with less than $c \cdot \frac{n}{\epsilon^2}$ samples for some sufficiently constant $c > 0$, and by Lemma~\ref{lem:pandora-lb-hellinger}, we have:
    \begin{align*}
        & \Pr \big[ \text{$A$ picks $\hypo \in \hypoclass^-$ given samples from $\vec{D^+}$} \big] \\
        \ge &\Pr \big[ \text{$A$ picks $\hypo \in \hypoclass^-$ given samples from $\vec{D^-}$} \big] - O(1)
        ~,
    \end{align*}
    for a sufficiently small constant inside the big-O notation.
    Putting together proves the lemma.
\end{proof}

As a direct corollary, we have the following via a simple counting argument.

\begin{lemma}
    \label{lem:pandora-lb-global-mistake-bound}
    There is an instance $\vec{D}$ for which the algorithm makes $\Omega(n)$ mistakes in expectation.
\end{lemma}

\begin{proof}
    By Lemma~\ref{lem:pandora-lb-local-mistake-bound}, if $\vec{D}$ is chosen from the $2^n$ possible instances uniformly at random, the algorithm makes a mistake on each box $i$ with constant probability.
    So the lemma follows.
\end{proof}

\begin{proof}[Proof of Theorem~\ref{thm:pandora} (Lower Bound)]
    Consider the instance in Lemma~\ref{lem:pandora-lb-global-mistake-bound}.
    Suppose the algorithm makes $\alpha n$ mistakes in expectation where $\alpha > 0$ is a constant.
    Then, by a standard probability argument, the probability that it makes at least $\frac{\alpha n}{2}$ mistakes is at least $\frac{1}{2}$.
    Hence, by Lemma~\ref{lem:pandora-lb-mistakes-and-error}, the expected additive error is at least $\Omega(\epsilon)$.
\end{proof}

\end{document}